\documentclass{article}

\usepackage{setspace,graphicx}     
\usepackage{amsmath,amsthm,amsfonts,amssymb,tikz,extarrows}
\usepackage{amsmath,amsfonts,xcolor,tikz,amsthm,xifthen,hyperref,graphicx}
\usepackage{hyperref}
\usepackage{amssymb}
\usepackage{amsmath}
\usepackage{braket}
\usepackage{tikz}
\usetikzlibrary{shapes,positioning}
\hypersetup{linkbordercolor={1 .8 .8},citebordercolor={.7 1 .7}}

\def\dprod{\mathop{\displaystyle \prod }}%

\newcommand{\D}{\mathcal D}

\newcommand{\expe}{\mathbb{E}}

\newcommand{\ind}{\mathbb{I}}

\newcommand{\W}{\mathcal W}

\newtheorem{theorem}{Theorem}
\newtheorem{corollary}[theorem]{Corollary}
\newtheorem{proposition}[theorem]{Proposition}
\newtheorem{lemma}[theorem]{Lemma}
\theoremstyle{definition}

\begin{document}

\title{On a causal quantum stochastic double product integral related to L%
\'{e}vy area }
\author{R L Hudson \\
Mathematics Department,\\
Loughborough University,\\
Loughborough,\\
Leicestershire LE11 3TU,\\
Great Britain. \and Y Pei \\
Mathematics Institute,\\
University of Warwick,\\
Coventry CV4 7AL,\\
Great Britain.}
\date{}
\maketitle

\section*{Abstract}
We study the family of causal double product integrals%
\begin{equation*}
  \prod_{a < x < y < b}\left(1 + i{\lambda \over 2}(dP_x dQ_y - dQ_x dP_y) + i {\mu \over 2}(dP_x dP_y + dQ_x dQ_y)\right)
\end{equation*}%
where $P$ and $Q$ are the mutually noncommuting momentum and position
Brownian motions of quantum stochastic calculus. The evaluation is motivated
heuristically by approximating the continuous double product by a discrete
product in which infinitesimals are replaced by finite increments. 
The latter is in turn approximated by the second quantisation of a discrete double product of rotation-like operators in different planes due to a result in \cite{hudson-pei15}. 
The main problem solved in this paper is the explicit evaluation of the continuum limit $W$ of the latter, and showing that $W$ is a unitary operator. 
The kernel of $W$ is written in terms of Bessel functions, and the evaluation is achieved by working on a lattice path model and enumerating linear extensions of related partial orderings, where the enumeration turns out to be heavily related to Dyck paths and generalisations of Catalan numbers.

\textit{AMS\ Subject Classification: }81S25, 05A15, 06A07

\textit{Keywords: }$\ $causal double product, L\'{e}vy's stochastic area,
position and momentum Brownian motions, linear extensions, Catalan numbers, Dyck paths.

\section{Introduction}

Following Volterra's philosophy of product integrals as continuous limits of discrete products \cite{slavik07}, quantum stochastic double product integrals
of \emph{rectangular} type have been constructed \cite{hudson-pei15}\ as limits of
discrete approximations obtained by replacing stochastic differentials by
discrete increments of the corresponding processes. Such constructions are
partially intuitive in character, involving nonrigorous manipulations of
unbounded operators. Nevertheless they can be shown to yield explicit
rigorously unitary operators which can then be shown in some
cases \cite{hudson-jones12}\ to satisfy the quantum stochastic differential equations
(qsde's) whose solutions provide the rigorous definition of the product
integral.

In this paper we initiate the much harder problem of constructing so-called 
\emph{causal} (or triangular) double product integrals in the same way,
first constructing discrete approximations by intuitive methods, which are
then shown rigorously to enjoy the property of unitarity, which will
allow rigorous verification of the qsde definitions.

The \emph{Fock space} $\mathcal{F}\left( \mathcal{H}\right) $ over a Hilbert
space $\mathcal{H}$ is conveniently defined \cite{parthasarathy92} as the Hilbert space
generated by the \emph{exponential vectors} $e\left( f\right) ,f\in \mathcal{%
H,}$ satisfying%
\begin{equation*}
\left\langle e\left( f\right) ,e\left( g\right) \right\rangle =\exp
\left\langle f,g\right\rangle ,\text{ }f,g\in \mathcal{H.} 
\end{equation*}

Rectangular product integrals live in the tensor product of two Fock spaces.
This form of "double" construction was originally motivated by its use to
construct explicit solutions of the quantum Yang-Baxter equation with a
given classical limit \cite{hudson07}, \cite{hudson-pulmannova05}, of purely algebraic
character as formal power series. From the analytic point of view, the
alternative causal constructs which are studied in the present paper which
live naturally in a single Fock space are of wider interest.

One example which we consider in some detail is closely related to L\'{e}%
vy's stochastic area \cite{levy51},\ and in particular to the L\'{e}vy area
formula for its characteristic function. In effect we replace the planar
Brownian motion \ by a quantum version in which the two components are the
mutually noncommuting momentum and position Brownian motions $P$ and $Q$ of
quantum stochastic calculus \cite{cockroft-hudson77}, which despite noncommutativity, can
be shown to be independent in a certain sense \cite{hudson13}. Other
noncommutative analogs of L\'{e}vy area are based on free probability \cite%
{capitaine-donatimartin00}; our own less radically noncommutative form is directly related to
physical applications \cite{hush-carvalho-hedges-james13}. It may also offer mathematically
significant relations, for example to Riemann zeta values through the links
to Euler and Bernoulli numbers \cite{ikeda-taniguchi10} of the classical L\'{e}vy area
formula. This is because, while the corresponding probability distribution
is the atomic one concentrated at zero, it deforms naturally to the
classical distribution at infinite temperature as the Fock "zero
temperature" momentum and position processes $P$ and $Q$ are deformed
through corresponding finite temperature processes \cite{chen-hudson13} to mutually
commuting independent Brownian motions.

We denote rectangular and causal product integrals by 
\begin{equation}
\dprod_{\lbrack a,b)\times \lbrack c,d)}{}\text{ }\left( 1+dr\right)
,\dprod_{<_{[a,b)}}\left( 1+dr\right)   \label{1}
\end{equation}
respectively, where $<_{[a,b)}$ is the set $\left\{ \left( x,y\right) \in 
\mathbb{R}^{2}:a\leq x\,<y<\,b\right\} .$ They are operators in the Hilbert
spaces $\mathcal{F}\left( L^{2}\left( [a,b)\right) \right) \otimes $ $%
\mathcal{F}\left( L^{2}\left( [c,d)\right) \right) $ and $\mathcal{F}\left(
L^{2}\left( [a,b)\right) \right) $ respectively. Both are characterised by
the \emph{generator} $dr$ which is a second rank tensor over the complex
vector space $\mathcal{I=}\mathbb{C}\left( dP,dQ,dT\right) $ of
differentials of the fundamental stochastic processes $P,$ $Q$ and the time
process $T$ of the calculus. They have rigorous definitions as solutions of
either forward or backward adapted quantum stochastic differential equations 
\cite{hudson14} in which $b$ or $a$ in \ref{1} is the time variable. They are
related by the \emph{coboundary relation}%
\begin{eqnarray*}
\dprod_{<_{[a,c)}}{}\left( 1+dr\right) &=&\left( \dprod_{<_{[a,b)}}{}\left(
1+dr\right) \otimes I\right) \dprod_{[a,b)\times \lbrack b,c)}{}\text{ }%
\left( 1+dr\right) \\
&&\left( I\otimes \dprod_{<_{[b,c)}}{}\left( 1+dr\right) \right)
\end{eqnarray*}%
in which the Fock space $\mathcal{F}\left( L^{2}\left( [a,c)\right) \right) $
is canonically split at time $b\in \lbrack a,c);$%
\begin{eqnarray*}
\mathcal{F}\left( L^{2}\left( [a,c)\right) \right) &=&\mathcal{F}\left(
L^{2}\left( [a,b)\right) \oplus L^{2}\left( [b,c)\right) \right) \\
&=&\mathcal{F}\left( L^{2}\left( [a,b)\right) \right) \otimes \mathcal{F}%
\left( L^{2}\left( [b,c)\right) \right)
\end{eqnarray*}%
allowing it to accommodate the operator $\dprod_{[a,b)\times \lbrack
b,c)}{} $ $\left( 1+dr\right) .$

A necessary and sufficient condition that they consist of unitary operators
is \cite{hudson14} that 
\begin{equation*}
dr+dr^{\dagger }+drdr^{\dagger }=0. 
\end{equation*}%
Here the space $\mathcal{I=}\mathbb{C}\left\langle dP,dQ,dT\right\rangle $
is equipped with the multiplication given by the quantum It\^{o} product rule%
\begin{equation*}
\begin{tabular}{c|c}
& $%
\begin{array}{ll}
\begin{array}{ll}
dP\text{ \ } & dQ%
\end{array}
& dT%
\end{array}%
$ \\ \hline
$%
\begin{array}{l}
\begin{array}{l}
dP \\ 
dQ%
\end{array}
\\ 
dT%
\end{array}%
$ & $%
\begin{array}{ll}
\begin{array}{ll}
dT & -idT \\ 
idT & \quad dT%
\end{array}
& 
\begin{array}{l}
0 \\ 
0%
\end{array}
\\ 
\begin{array}{ll}
\quad 0 & \quad \quad 0%
\end{array}
& \text{ }0%
\end{array}%
$%
\end{tabular}%
\text{ } 
\end{equation*}%
and $\mathcal{I\otimes I}$ with the corresponding tensor product
multiplication, together with the natural involution $\dagger $ derived from
the self-adjointness of $P,$ $Q$ and $T.$

Two examples of such unitary generators are%
\begin{eqnarray*}
dr_{1} &=&i\left( dP\otimes dQ-dQ\otimes dP\right) , \\
dr_{2} &=&i\left( dP\otimes dP+dQ\otimes dQ\right) .
\end{eqnarray*}%
$dr_{1}$ relates to quantum L\'{e}vy area in which the independent
one-dimensional component Brownian motions of planar Brownian motion are
replaced by $P$ and $Q.$ In the same spirit, $dr_{2}$ relates to a quantum
version of the Bessel process, the radial part of planar Brownian motion.
The general form of unitary generator in which the time differential $dT$
does not appear is \cite{hudson-pei15}\ the real linear combination%
\begin{equation*}
  dr_{\lambda ,\mu }={\lambda \over 2} dr_{1}+{\mu \over 2} dr_{2}. 
\end{equation*}

In this paper we begin the explicit construction of the unitary causal double product integral 
\begin{align*}
  E := \prod_{<_{[a,b)}} (1 + dr_{\lambda, \mu})
\end{align*}
as the second quantisation $\Gamma \left( W\right) $ of a unitary
operator $W$ which differs from the identity operator $I$ by an integral
operator on the Hilbert space $L^{2}\left( [a,b)\right) $ whose kernel will
be found explicitly.

\textit{Acknowledgements: }Parts of this work were completed when the 
Robin Hudson visited the Mathematics Department of Chungbuk National University in
Korea, whose warm hospitality is gratefully acknowledged, along with
conversations with John Gough, Paul Jones and Janosch Ortmann.
Parts of this work were completed when Yuchen Pei visited the School of Mathematics of Trinity College Dublin, Mathematical Research and Conference Center of the Institute of Mathematics of the Polish Academy of Sciences, Chungbuk National University and Gyeongbokgung. Conversations with Neil O'Connell is also acknowledged. The research of Yuchen Pei is supported by EPSRC grant number EP/H023364/1.

\section{The L\'evy stochastic area}
Before moving on to construct $E$, let us take a detour and explain the motivation of this problem.

The stochastic L\'evy area introduced in \cite{levy51} is defined as the signed area formed by connecting the endpoints of a 2-dimensional Brownian path.
More specifically, it is defined as
\begin{align*}
  L = {1 \over 2}\int_{0 \le s_1 < s_2 < t} d B^1_{s_1} d B^2_{s_2} - d B^2_{s_1} d B^1_{s_2} d s_1 d s_2,
\end{align*}
where $B^1$ and $B^2$ are two independent Brownian motions.
The L\'evy area formula shows the characteristic function of $L$:
\begin{align}
  \expe\prod_{0 \le s_1 < s_2 < t} \left( 1 + {i \lambda \over 2} (dB^1_{s_1} dB^2_{s_2} - dB^2_{s_1} dB^2_{s_2}) \right) = \expe e^{i \lambda L} = \text{sech} {\lambda t \over 2}.
  \label{eq:laformula}
\end{align}

The L\'evy area formula has many interesting connotations. For example there are connections to integrable systems, Bernoulli and Euler polynomials, and hence to the values of the Riemann zeta function \cite{yor80}.
For some recent work and further references see \cite{ikeda-taniguchi10,ikeda-taniguchi11}.
Also, to within normalisation and rescaling it is equal to its Fourier transform, the density of the corresponding probability distribution, which is a boundary point of the Meixner family \cite{meixner34}.

Noncommutative analogues of L\'evy area have been previously considered in free probability~\cite{capitaine-donatimartin00,ortmann13,victoir04}.
Also in this connection Deya and Schott~\cite{deya-schott13} emphasise the primacy of iterated stochastic integrals which accords with our philosophy.
But in this paper we are concerned with a noncommutative analogue of a more conservative kind which arises in quantum stochastic calculus~\cite{hudson-parthasarathy84,parthasarathy92}, regarded as a noncommutative extension, rather than a radically noncommutative analogue, of It\=o calculus.
This allows a very natural variant of the area to be constructed using the minimal one-dimensional version of the calculus.
It may be regarded as a response to the call~\cite{applebaum10} for a study in this quantum context of some of the deeper properties of Brownian motion, as well as a furtherance of the theory of quantum stochastic product integrals~\cite{hudson07,hudson07a,hudson-ion81}.

By replacing $B^1$ and $B^2$ with $P$ and $Q$, the iterated quantum stochastic integral
\begin{align*}
  K(t) = {1 \over 2} \int_{0 \le x < y < t} (dP_x dQ_y - dQ_x dP_y)
\end{align*}
has some interesting properties~\cite{hudson13,chen-hudson13}. For example it is evidently invariant under gauge transformations, which replace $(P, Q)$ by $(P^\theta, Q^\theta)$  where
\begin{align*}
  P^\theta = P \cos \theta - Q \sin \theta,\qquad Q^\theta = P \sin \theta + Q \cos \theta;
\end{align*}
equivalently the corresponding creation and annihilation processes are multiplied by $e^{\pm i \theta}$.
In particular, taking $\theta = - {\pi \over 2}$ it is invariant under the replacement $(P, Q)$ by $(Q, - P)$.
Thus, unlike the separate processes $P$ and $Q$, it can be canonically ``rolled'' onto a (one-dimensional) Riemannian manifold, and its multidimensional version~\cite{friz-victoir10} can similarly be rolled onto a multidimensional manifold, with possible applications to quantum stochastic proofs of index theorems, by identifying the canonical Brownian motion on the manifold generated by the Laplacian as $P^\theta$ with arbitrarily chosen $\theta$.

It can also be verified~\cite{chen-hudson13} that all moments of $K(t)$ vanishes in the vacuum state, so that $K(t)$ vanishes in a probabilistic sense, even though it is not the zero operator.

But it is not $K$ which is the main object of study.
Because $\exp(a + b) \neq \exp a \exp b$ when $a$ and $b$ do not commute, the exponential
\begin{align*}
  \exp(i \lambda K(t)) = \exp\left( {i \lambda \over 2} \int_{0 \le x < y < t} (dP_x dQ_y - dQ_x dP_y) \right)
\end{align*}
does not reflect in a coherent way the continuous tensor product structure underlying the quantum stochastic calculus.
Thus, motivated by the hope of finding quantum extensions of, in particular, the L\'evy area formula \eqref{eq:laformula}, and associated relations with Euler and Bernoulli polynomials~\cite{ikeda-taniguchi11} we investigate the double product integral
\begin{align*}
  \prod_{<_{[a,b)}} (1 + d r_1) = \prod_{a \le x < y < b} \left( 1 + {i \lambda \over 2} (dP_x dQ_y - dQ_x dP_y) \right)
\end{align*}

However, as it turns out, the more general object
\begin{align*}
  E &= \prod_{<_{[a,b)}} (1 + d r_{\lambda, \mu}) \\
  &= \prod_{a \le x < y < b} \left( 1 + {i \lambda \over 2} (dP_x dQ_y - dQ_x dP_y) + {i \mu \over 2} (dP_x dP_y + dQ_x dQ_y) \right).
\end{align*}
is more fundamental and, surprisingly, simpler to study.

\section{A discrete double product of unitary matrices}

The first stage of the construction of $E$ is similar to that of the rectangular
case construction outlined in \cite{hudson-pei15}, in that we approximate $%
\dprod_{<_{[a,b)}}{}\left( 1+dr_{\lambda ,\mu }\right) $ by a discrete
double product $\Pi _{1\leq j<k\leq N}\left( I+\delta _{N}^{j,k}r_{\lambda
,\mu }\right) ,$ where $\delta _{N}^{j,k}r_{\lambda ,\mu }$ is obtained from 
$dr_{\lambda ,\mu }$ by replacing each basic differential $dX\in \left\{
dP,dQ\right\} $ contributing to $dr_{\lambda ,\mu }\in \mathcal{I}\otimes 
\mathcal{I}$ in the first copy of $\mathcal{I}$ by the $j$-th increment $%
X_{x_{j}} -X_{x_{j-1}} $ and in the second copy of $%
\mathcal{I}$ by the $k$-th increment $X_{x_{k}} -X_{x_{k-1}}$
over the equipartition 
\begin{equation*}
\lbrack a,b)=\sqcup _{j=1}^{N}[x_{j-1},x_{j}),\text{ }x_{j}=a+\frac{j}{N}%
\left( b-a\right) =: a + j \Delta_N. 
\end{equation*}%
Thus, for example,%
\begin{eqnarray*}
  \delta _{N}^{j,k}r_{1} &=&{i \over 2}\left( \left( P_{ x_{j}} -P_{ x_{j-1}} \right) \otimes \left( Q_{ x_{k}} -Q_{x_{k-1}} \right) \right. \\
&&-\left. \left( Q_{x_{j}} -Q_{x_{j-1}} \right) \otimes \left( P_{ x_{k}} -P_{ x_{k-1}} \right) \right) .
\end{eqnarray*}%
Introducing the standard canonical pairs $\left( p_{j},q_{j}\right)
,j=1,2,...,N,$ given by%
\begin{equation*}
p_{j}=\sqrt{\frac{b-a}{N}}\left( P\left( x_{j}\right) -P\left(
x_{j-1}\right) \right) ,\text{ }q_{j}=\sqrt{\frac{b-a}{N}}\left( Q\left(
x_{j}\right) -Q\left( x_{j-1}\right) \right) , 
\end{equation*}%
which satisfy the canonical commutation relations%
\begin{equation}
\left[ p_{j},q_{k}\right] =-2i\delta _{j,k},\left[ p_{j},p_{k}\right] =\left[
q_{j},q_{k}\right] =0,\text{ }  \label{7}
\end{equation}%
we write%
\begin{equation*}
\delta _{N}^{j,k}r_{1}=i\frac{b-a}{2 N}\left( p_{j}q_{k}-q_{j}p_{k}\right) 
\end{equation*}%
and more generally%
\begin{equation*}
\delta _{N}^{j,k}r_{\lambda ,\mu }=i\frac{b-a}{2 N}\left( \lambda \left(
p_{j}q_{k}-q_{j}p_{k}\right) +\mu \left( p_{j}p_{k}+q_{j}q_{k}\right)
\right) . 
\end{equation*}%
Our approximation is thus%
\begin{equation}
  \begin{aligned}
\dprod_{<_{[a,b)}}{}\left( 1+dr_{\lambda ,\mu }\right) &\simeq &\Pi _{1\leq
j<k\leq N}\left( I+i\frac{b-a}{2 N}\left( \lambda \left(
p_{j}q_{k}-q_{j}p_{k}\right) +\mu \left( p_{j}p_{k}+q_{j}q_{k}\right)
\right) \right)  \\
&\simeq &\Pi _{1\leq j<k\leq N}\exp \left( i\frac{b-a}{2 N}\left( \lambda
\left( p_{j}q_{k}-q_{j}p_{k}\right) +\mu \left( p_{j}p_{k}+q_{j}q_{k}\right)
\right) \right)  
  \end{aligned}
  \label{6}
\end{equation}%
for large $N$.

Temporarily let us fix $j<k$ and write $\left( p,q\right) =\left(
p_{j},q_{j}\right) ,$ $\left( p^{\prime },q^{\prime }\right) =\left(
p_{k},q_{k}\right) $ so that%
\begin{equation}
\left[ p,q\right] =-2i,\text{ }\left[ p^{\prime },q^{\prime }\right] =-2i,%
\text{ }\left[ p,q^{\prime }\right] =\left[ q,p^{\prime }\right] =\left[
p,p^{\prime }\right] =\left[ q,q^{\prime }\right] =0.\text{ }  \label{2}
\end{equation}%
We recall \cite{parthasarathy92} that, for an arbitrary Hilbert space $\mathcal{H}$ and
vector $f\in \mathcal{H}$ the corresponding \emph{Weyl operator} $W\left(
f\right) $ is the unique unitary operator on $\mathcal{F}\left( \mathcal{H}%
\right) $ which acts on each exponential vector $e\left( g\right) ,$ $g\in 
\mathcal{H}$ as%
\begin{equation*}
W\left( f\right) e\left( g\right) =e^{-%
{\frac12}%
\left\Vert f\right\Vert^2 -\left\langle f,g\right\rangle }e(f+g).
\end{equation*}%
The Weyl operators satisfy the Weyl relation%
\begin{equation}
W\left( f\right) W\left( g\right) =e^{-i\text{Im}\left\langle
f,g\right\rangle }W\left( f+g\right) .  \label{4}
\end{equation}%
A convenient rigorous realisation of two canonical pairs satisfying the
commutation relations (\ref{2}) can be constructed in terms of the
one-parameter unitary groups of which they are the self-adjoint
infinitesimal generators, which are Weyl operators on the Fock space $%
\mathcal{F}\left( \mathbb{C}^{2}\right) $ over $\mathbb{C}^{2}.$ Regarding $%
\mathbb{C}^{2}$ as a space of column vectors, we take 
\begin{eqnarray*}
e^{ixp} &=&W\left( \left( x,0\right) ^{\tau }\right) ,e^{ixq}=W\left( \left(
-ix,0\right) ^{\tau }\right) , \\
e^{ixp^{\prime }} &=&W\left( \left( 0,x\right) ^{\tau }\right)
,e^{ixq^{\prime }}=W\left( \left( 0,-ix\right) ^{\tau }\right) .
\end{eqnarray*}%
for arbitrary $x\in \mathbb{R,}$ noting that these four families of Weyl
operators are indeed one-parameter unitary groups, and that the commutation
relations (\ref{2}) follow by parametric differentiation, for example from
the relations%
\begin{eqnarray*}
W\left( \left( x,0\right) ^{\tau }\right) W\left( \left( -iy,0\right) ^{\tau
}\right)  &=&e^{2ixy}W\left( \left( -iy,0\right) ^{\tau }\right) W\left(
\left( x,0\right) ^{\tau }\right) , \\
W\left( \left( 0,x\right) ^{\tau }\right) W\left( \left( 0,-iy\right) ^{\tau
}\right)  &=&e^{2ixy}W\left( \left( 0,-iy\right) ^{\tau }\right) W\left(
\left( 0,x\right) ^{\tau }\right) ,
\end{eqnarray*}%
all of which are consequences of (\ref{4}).

Theorem 1 below, which is proved in \cite{hudson-pei15}, gives a corresponding
rigorous explicit form of the self-adjoint operator 
\begin{equation*}
L\left( \lambda ,\mu \right) =\lambda \left( pq^{\prime }-qp^{\prime
}\right) +\mu \left( pp^{\prime }+qq^{\prime }\right) 
\end{equation*}%
in this realisation. Before stating it we recall \cite{parthasarathy92} that the \emph{%
second quantisation }of a unitary operator $U$ on a Hilbert space $\mathcal{H%
}$ is the unique unitary operator $\Gamma \left( U\right) $ on $\mathcal{F}%
\left( \mathcal{H}\right) $ which acts on the exponential vectors as 
\begin{equation*}
\Gamma \left( U\right) e\left( f\right) =e\left( Uf\right) .
\end{equation*}%
It is related to the Weyl operators by%
\begin{equation}
\Gamma \left( U\right) W\left( f\right) =W\left( Uf\right) \Gamma \left(
U\right)   \label{3}
\end{equation}%
for arbitrary $f\in \mathcal{H.~}$Second quantisation is multiplicative, in
the sense that%
\begin{equation}
\Gamma \left( U_{1}U_{2}\right) =\Gamma \left( U_{1}\right) \Gamma \left(
U_{2}\right)   \label{5}
\end{equation}%
for arbitrary unitary $U_{1},U_{2}.$

\begin{theorem}\label{t:hp}
$L\left( \lambda ,\mu \right) $ generates the one-parameter unitary group%
\begin{equation*}
e^{ixL\left( \lambda ,\mu \right) }=\Gamma \left( \left[ 
\begin{array}{ll}
\cos \left( 2x\left\vert \nu\right\vert \right)  & -e^{-i\phi }\sin \left(
2x\left\vert \nu\right\vert \right)  \\ 
e^{i\phi }\sin \left( 2x\left\vert \nu\right\vert \right)  & \cos \left(
2x\left\vert \nu\right\vert \right) 
\end{array}%
\right] \right) ,\text{ }x\in \mathbb{R}
\end{equation*}%
where $\nu =\lambda +i\mu =e^{i\phi }\left\vert \nu \right\vert $ and the
matrix operates on column vectors in $\mathbb{C}^{2}$ by multiplication on
the left.
\end{theorem}

We now use Theorem \ref{t:hp} to construct an explicit second quantisation of the
approximation (\ref{6}).

Let us first construct a different realisation of the canonical pairs $%
\left( p_{j},q_{j}\right) ,j=1,2,...,n,$ satisfying (\ref{7}) in the Fock
space $\mathcal{F}\left( \mathbb{C}^{n}\right) $ over $\mathbb{C}^{n},$ by
defining%
\begin{equation*}
e^{ixp_{j}}=W\left( x\varepsilon _{j}\right) ,e^{ixq_{j}}=W\left(
-ix\varepsilon _{j}\right) 
\end{equation*}%
where $\left( \varepsilon _{j}\right) _{j=1}^{n}$ is the standard
orthonormal basis of $\mathbb{C}^{n}$, $\varepsilon _{j}=\left( 0,...,%
\overset{\left( j\right) }{1},0,...,0\right) ^{\tau }.$ Correspondingly, in
view of Theorem \ref{t:hp}, each operator 
\begin{equation*}
\exp \left( i\frac{b-a}{2 N}\left( \lambda \left( p_{j}q_{k}-q_{j}p_{k}\right)
+\mu \left( p_{j}p_{k}+q_{j}q_{k}\right) \right) \right) 
\end{equation*}
is realised as the second quantisation $\Gamma(R^N_{j, k})$ where%
\begin{equation*}
  R^N_{j,k}:=
\left[ 
\begin{array}{cccccccc}
&  &  & (j) &  & (k) &  &  \\ 
& 1 & \cdots & 0 & \cdots & 0 & \cdots & 0 \\ 
& \vdots & \ddots & \vdots & \cdots & \vdots & \cdots & \vdots \\ 
(j) & 0 & \cdots & \cos \left( \frac{ \left( b-a\right) }{N} |\nu| \right)
& \cdots & -{\bar \nu \over |\nu|} \sin \left( \frac{ \left( b-a\right) }{N} |\nu| \right) & \cdots
& 0 \\ 
& \vdots & \cdots & \vdots & \ddots & \vdots & \cdots & \vdots \\ 
(k) & 0 & \cdots & {\bar \nu \over |\nu|} \sin \left( \frac{ \left( b-a\right) }{N} |\nu| \right)
& \cdots & \cos \left( \frac{ \left( b-a\right) }{N} |\nu| \right) & \cdots
& 0 \\ 
& \vdots & \cdots & \vdots & \cdots & \vdots & \ddots & \cdots \\ 
& 0 & \cdots & 0 & \cdots & 0 & \cdots & 1%
\end{array}%
\right] .
\end{equation*}
In view of the multiplicativity property (\ref{5}) the discrete double
product (\ref{6}) is correspondingly realised as the second quantisation of
the product%
\begin{align}
  \prod_{1 \le j < k \le N} R^N_{j, k}.
  \label{8}
\end{align}
We now embed the matrix (\ref{8}) as a unitary operator $\W_{N}$ on $%
L^{2}\left( [a,b[\right) $ by mapping the standard basis of $%
\mathbb{C}^{N}$ to the orthonormal family $\left( \chi _{1},\chi
_{2},...,\chi _{N}\right) $ of normalized indicator functions%
\begin{equation*}
  \chi _{j}\left( x\right) = \sqrt{ N \over b - a} \ind_{[x_{j - 1}, x_j)}.
\end{equation*}%
By definition $\W_{N}$ acts as the identity operator $I$ on $\left( \chi
_{1},\chi _{2},...,\chi _{N}\right) ^{\perp }.$

Our objective in the remainder of this paper is to find an explicit form for
the (weak) limit 
\begin{equation*}
W=\underset{}{\lim_{N\rightarrow \infty }}\W_{N} 
\end{equation*}%
and to prove that $R$ is unitarity. The corresponding problems for
rectangular unitary product integrals was solved in outline in \cite{hudson-pei15}.
The causal case considered here is considerably more difficult, because the
method of iterated limits which reduces the rectangular case to a double
application of the time-orthogonal unitary dilation of \cite{hudson-ion-parthasarathy82}, is not
applicable. Instead a combinatorial argument based on a lattice path model
is used. For a similar alternative approach, avoiding the iterated limit
technique, to the rectangular product in the particular case of the
generator $dr_{1}$ corresponding to the quantum L\'{e}vy area, see \cite%
{hudson-jones12}; however the combinatorics for the rectangular case is much simpler
than here and it has no direct relation to L\'{e}vy area.

\section{A lattice path model and linear extensions of partial orderings}
So we want to calculate the limit of the triangular double product of $N\times N$ matrices
\begin{equation}
  \W_{N}=\prod_{1\leq j<k\leq N} R^N_{j,k}.
\label{W}
\end{equation}%
Here, for elements $x_{j,k}$ of an associative algebra having the
property that $x_{j,k}$ commutes with $x_{j^{\prime }k^{\prime }}$ whenever
both $j\neq j^{\prime }$ and $k\neq k^{\prime }$ we define the ordered
double product $\prod_{1\leq j<k\leq N}x_{j,k}$ by any of the equivalent
prescriptions 
\begin{equation*}
\prod_{1\leq j<k\leq N}x_{j,k}=\prod_{j=1}^{N-1}\left[
\prod_{k=j+1}^{N}x_{j,k}\right] =\prod_{k=2}^{N}\left[
\prod_{j=1}^{k-1}x_{j,k}\right] =\prod_{r=1}^{\frac{1}{2}%
N(N-1)}x_{j_{r},k_{r}}
\end{equation*}%
where $\left( \left( j_{1},k_{1}\right) ,\left( j_{2},k_{2}\right)
,...,\left( j_{\frac{1}{2}N(N-1)},k_{\frac{1}{2}N(N-1)}\right) \right) $ is
any ordering of the $\frac{1}{2}N(N-1)$ pairs $(j,k),1\leq j<k\leq N$
\medskip which is \emph{allowed,} meaning that it has the property that %
\begin{equation}
  (j_{r},k_{r})\text{ precedes }( j_{s},k_{s})\text{ if both } j_{r}\leq j_{s} \text{ and }k_{r}\leq k_{s}.
  \label{eq:ncprodcondition}
\end{equation}

In constructing the limit as $N\rightarrow \infty $ we use the
small angle approximations for sine and cosine, so that 
\begin{align*}
\begin{pmatrix}
\cos {b - a \over N} |\nu| & - {\bar \nu \over |\nu|} \sin{b - a \over N} |\nu| \\
{\nu \over |\nu|} \sin{b - a \over N} |\nu| & \cos {b - a \over N} |\nu| \\
\end{pmatrix}
= I + {b - a \over N}
\begin{pmatrix}
0 & -\bar\nu \\ \nu & 0
\end{pmatrix}
+ O(N^{-2})
\end{align*}
hence
\begin{align*}
  R^N_{j,k} = I + {(b - a) \over N} (-\bar \nu \ket{\chi_j} \bra{\chi_k} + \nu \ket{\chi_k} \bra{\chi_j}) + O(N^{-2}).
\end{align*}
When there is no ambiguity, for any integers $j$ and $k$, we use abbreviations $|j\rangle := | \chi_j \rangle$ and $\langle k | := \langle \chi_k |$.
Then the product (\ref{W}) becomes%
\begin{align*}
\W_{N} \simeq \prod_{1\leq j<k\leq N}\left( I+\frac{\left( b-a\right) }{N} Z^{N}(j,k)\right) =: W_N
\end{align*}%
where 
\begin{equation*}
Z^{N}(j,k)=-\bar \nu \left\vert j\right\rangle \left\langle k\right\vert + \nu \left\vert k\right\rangle \left\langle j\right\vert.
\end{equation*}

To compute this, we introduce and work on a lattice path model.
Consider a lattice $L_s := \{(m, n): 1 \le m \le s, 0 \le n \le 1\}$. 
We call $(m, 1)_{1 \le m \le s}$ the upper vertices, and $(m, 0)_{1 \le m \le s}$ the lower vertices.
Denote by $\Pi_s$ the set of lattice path $\pi = (m_i, b_i)_{i = 1}^s$ satisfying the following two conditions:
    \begin{enumerate}
      \item $m_i = i$ for $i = 1, \dots, s$
      \item there does not exist an $i$ such that $b_i = b_{i+1} = 0$
    \end{enumerate}
For convenience, we write $\pi(i) = b_i$ and let $\pi = (\pi(i))_i$.
We call any $\pi \in \Pi_s$ a path of length $s - 1$.
It is straightforward to verify by induction that
\begin{align*}
  |\Pi_s| = \text{Fib}_{s + 2} = {\Phi^{s + 2} - (- \Phi)^{- s - 2} \over \sqrt 5},
\end{align*}
where $\text{Fib}_n$ is the $n$th Fibonacci number and $\Phi$ is the golden ratio ${\sqrt 5 + 1 \over 2}$.

    If we assign weight $\theta(v)$ to each vertex $v$ in $L_s$, then we can define the weight $\theta(\pi)$ of a path $\pi \in \Pi_s$ by the product of the weights of its vertices:
    \begin{align*}
      \theta(\pi) := \prod_{i=1}^s \theta(i, \pi(i)).
    \end{align*}

  For any $s$-array of pairs $\{p_{ij}: 1 \le i \le s, 1 \le j \le 2\}$, define its associated weight $\theta_p(v)$ for any $v = (m, b) \in A$ to be 
  \begin{equation*}
    \theta_p (v) = 
    \begin{cases}
      \nu \ket{p_{i2}} \bra{p_{i1}},&\text{ if } b = 0;\\
      -\bar \nu \ket{p_{i1}}\bra{p_{i2}}, &\text{ if } b = 1.
    \end{cases}
  \end{equation*}
  Finally, define the weight $\theta_p(\pi)$ of a path in the same way as before.

For example, if we label the vertices by their weights associated to $p$, then the following is a path of $\Pi_5$:
\begin{center}
  \begin{tikzpicture}[xscale=2.5,yscale=1]
    \path (.8,1) node (p12) {$+\nu\ket{p_{12}}\bra{p_{11}}$};
    \path (.8,2) node (p11) {$-\bar\nu\ket{p_{11}}\bra{p_{12}}$};
    \path (2,1) node (p22) {$+\nu\ket{p_{22}}\bra{p_{21}}$};
    \path (2,2) node (p21) {$-\bar\nu\ket{p_{21}}\bra{p_{22}}$};
    \path (3,1) node (p32) {$+\nu\ket{p_{32}}\bra{p_{31}}$};
    \path (3,2) node (p31) {$-\bar\nu\ket{p_{31}}\bra{p_{32}}$};
    \path (4.2,1) node (ps-12) {$+\nu\ket{p_{4,2}}\bra{p_{4,1}}$};
    \path (4.2,2) node (ps-11) {$-\bar\nu\ket{p_{4,1}}\bra{p_{4,2}}$};
    \path (5.4,1) node (ps2) {$+\nu\ket{p_{5,2}}\bra{p_{5,1}}$};
    \path (5.4,2) node (ps1) {$-\bar\nu\ket{p_{5,1}}\bra{p_{5,2}}$};
    \draw (p11)--(p21)--(p32)--(ps-11)--(ps1);
  \end{tikzpicture}
\end{center}
but not the following because the third edge connects two bottom vertices:
\begin{center}
  \begin{tikzpicture}[xscale=2.5,yscale=1]
    \path (.8,1) node (p12) {$+\nu\ket{p_{12}}\bra{p_{11}}$};
    \path (.8,2) node (p11) {$-\bar\nu\ket{p_{11}}\bra{p_{12}}$};
    \path (2,1) node (p22) {$+\nu\ket{p_{22}}\bra{p_{21}}$};
    \path (2,2) node (p21) {$-\bar\nu\ket{p_{21}}\bra{p_{22}}$};
    \path (3,1) node (p32) {$+\nu\ket{p_{32}}\bra{p_{31}}$};
    \path (3,2) node (p31) {$-\bar\nu\ket{p_{31}}\bra{p_{32}}$};
    \path (4.2,1) node (ps-12) {$+\nu\ket{p_{4,2}}\bra{p_{4,1}}$};
    \path (4.2,2) node (ps-11) {$-\bar\nu\ket{p_{4,1}}\bra{p_{4,2}}$};
    \path (5.4,1) node (ps2) {$+\nu\ket{p_{5,2}}\bra{p_{5,1}}$};
    \path (5.4,2) node (ps1) {$-\bar\nu\ket{p_{5,1}}\bra{p_{5,2}}$};
    \draw (p11)--(p21)--(p32)--(ps-12)--(ps1);
  \end{tikzpicture}
\end{center}

Any $s$-array of pairs $p = (p_{ij})_{1 \le i \le s, 1 \le j \le 2}$ satisfying the following condition
\begin{equation}
  \mathbb I_{p_{i, 1} = p_{i + 1, 1}} + \mathbb I_{p_{i, 2} = p_{i + 1, 2}} + \mathbb I_{p_{i, 2} = p_{i + 1, 1}} = 1,\qquad p_{i, 1} \neq p_{i + 1, 2}
  \label{eq:boolean}
\end{equation}
can be associated with a path $\pi_p \in \Pi_s$ in the following way:
\begin{equation*}
  (\pi(i), \pi(i + 1)) = 
  \begin{cases}
    (0, 1),&\text{ if }p_{i, 1} = p_{i + 1, 1}\\
    (1, 1),&\text{ if }p_{i, 2} = p_{i + 1, 1}\\
    (1, 0),&\text{ if }p_{i, 2} = p_{i + 1, 2}
  \end{cases}.
\end{equation*}
Note that this is equivalent to
\begin{align*}
  \theta_p(\pi_p) = \prod_{i = 1}^s Z^N(p_{i, 1}, p_{i, 2}).
\end{align*}
\begin{lemma}\label{l:path}
    \begin{align*}
      W_N = I + \sum_{s = 1}^{N(N - 1) / 2}\tilde{w}_{s,N},
    \end{align*}
    where
    \begin{align*}
      \tilde w_{s,N} = \left( {b - a \over N} \right)^s \sum_{(*)} \theta_p(\pi_p)
    \end{align*}
where the domain $(*)$ of the summation is
\begin{align*}
  &(1 \le p_{i, j} \le N, \forall 1 \le i \le s, 1 \le j \le 2)\text{ AND }\\
  &\qquad(p_{i, 1} < p_{i + 1, 1} < p_{i, 2} = p_{i + 1, 2} \text{ OR} \\
  &\qquad p_{i, 1} < p_{i, 2} = p_{i + 1, 1} < p_{i + 1, 2} \text{ OR} \\
  &\qquad p_{i, 1} = p_{i + 1, 1} < p_{i, 2} < p_{i + 1, 2}, \forall 1 \le i \le s - 1)
\end{align*}
\end{lemma}
\begin{proof}
  For any rearrangement $(j_i, k_i)_{1 \le i \le {N (N - 1) \over 2}}$ of $\{1 \le j < k \le N\}$ satisfying \eqref{eq:ncprodcondition},
  \begin{align*}
    \prod_{1 \le j < k \le N} (I + Z(j, k)) = \prod_{i = 1}^{N (N - 1) / 2}(I + Z(j_i, k_i)) = I + \sum_{s = 1}^{N (N - 1) / 2} \sum_{(**)} \prod_{r = 1}^s Z(p_{r1}, p_{r2}),
  \end{align*}
  where domain $(**)$ is
  \begin{align*}
    (p_{11}, p_{12}), &(p_{21}, p_{22}), \dots, (p_{s1}, p_{s2}) \\
    &\text{ is a subsequence of } (j_1, k_1), (j_2, k_2), \dots, (j_{N (N - 1) / 2}, k_{N (N - 1) / 2}).
  \end{align*}
  Now for the product $\prod_{r = 1}^s Z(p_{r1}, p_{r2})$ to be nonzero, the product of $Z(p_{i,1}, p_{i,2}) \times Z(p_{i + 1, 1}, p_{i + 1, 2})$ for each $i$ has to be nonzero, that is
  \begin{align*}
    (- \ket{p_{i,1}} \bra{p_{i, 2}} + \ket{p_{i, 2}} \bra{p_{i, 1}})
    (- \ket{p_{i + 1,1}} \bra{p_{i + 1, 2}} + \ket{p_{i + 1, 2}} \bra{p_{i + 1, 1}}) \neq 0.
  \end{align*}
  This in turn is equivalent to
  \begin{align*}
    (p_{i, 2} = p_{i + 1, 2}) \vee (p_{i, 2} = p_{i + 1, 1}) \vee (p_{i, 1} = p_{i + 1, 1}) \vee (p_{i, 1} = p_{i + 1, 2}).
  \end{align*}
We analyse these four possibilities one by one.
\begin{enumerate}
  \item If $p_{i, 2} = p_{i + 1, 2}$, then by \eqref{eq:ncprodcondition}, and since $(p_{i, 1}, p_{i, 2}) \neq (p_{i + 1, 1}, p_{i + 1, 2})$, only when $p_{i, 1} < p_{i + 2, 1}$ can the product be nonzero. In this case the coordinates are ordered as $p_{i, 1} < p_{i + 1, 1} < p_{i, 2} = p_{i + 1, 2}$.
  \item If $p_{i, 2} = p_{i + 1, 1}$, then since $p_{i, 1} < p_{i, 2} = p_{i + 1, 1}$ and $p_{i, 2} = p_{i + 1, 1} < p_{i + 1, 2}$, we have that \eqref{eq:ncprodcondition} is satisfied. Therefore this case is also included / permitted in the product. The ordering of the coordinates is $p_{i, 1} < p_{i, 2} = p_{i + 1, 1} < p_{i + 1, 2}$.
  \item If $p_{i, 1} = p_{i + 1, 1}$, then similar to Case 1, the coordinates have to satisfy $p_{i, 1} = p_{i + 1, 1} < p_{i, 2} < p_{i + 1, 2}$ for the product to be nonzero.
  \item If $p_{i, 1} = p_{i + 1, 2}$, then $p_{i, 1} = p_{i + 1, 2} > p_{i, 1}$ and $p_{i, 2} > p_{i, 1} = p_{i + 1, 2}$ violates \eqref{eq:ncprodcondition}, hence this case never happens.
\end{enumerate}
The three feasible cases are illustrated as below.
\tikzstyle{every node}=[circle, draw, fill=black!50, inner sep=0pt, minimum width=4pt]
\begin{equation}
\begin{aligned}
  i\;\;\;\;i+1&&i\;\;\;\;i+1&&i\;\;\;\;i+1\\
\begin{tikzpicture}
  \draw (1,1)node{}--(2,0)node{};
\end{tikzpicture}&&
\begin{tikzpicture}
  \draw (1,1)node{}--(2,1)node{};
  \draw[white] (0,0)--(1,0);
\end{tikzpicture}&&
\begin{tikzpicture}
  \draw (1,0)node{}--(2,1)node{};
\end{tikzpicture}\\
\tikzstyle{every node}=[]
\begin{tikzpicture}
  \node(11) at (1,1) {$p_{i1}$};
  \node(10) at (1,0) {$p_{i2}$};
  \node(21) at (2,1) {$p_{i+1,1}$};
  \node(20) at (2,0) {$p_{i+1,2}$};
  \draw[white](11)--node[black]{$<$}(21);
  \draw[white](11)--node[black,sloped]{$<$}(10);
  \draw[white](21)--node[black,sloped]{$<$}(20);
  \draw[white](10)--node[black]{$=$}(20);
\end{tikzpicture}&&
\tikzstyle{every node}=[]
\begin{tikzpicture}
  \node(11) at (1,1) {$p_{i1}$};
  \node(10) at (1,0) {$p_{i2}$};
  \node(21) at (2,1) {$p_{i+1,1}$};
  \node(20) at (2,0) {$p_{i+1,2}$};
  \draw[white](11)--node[black]{$<$}(21);
  \draw[white](11)--node[black,sloped]{$<$}(10);
  \draw[white](21)--node[black,sloped]{$<$}(20);
  \draw[white](10)--node[black,sloped]{$=$}(21);
  \draw[white](10)--node[black]{$<$}(20);
\end{tikzpicture}&&
\tikzstyle{every node}=[]
\begin{tikzpicture}
  \node(11) at (1,1) {$p_{i1}$};
  \node(10) at (1,0) {$p_{i2}$};
  \node(21) at (2,1) {$p_{i+1,1}$};
  \node(20) at (2,0) {$p_{i+1,2}$};
  \draw[white](11)--node[black]{$=$}(21);
  \draw[white](11)--node[black,sloped]{$<$}(10);
  \draw[white](21)--node[black,sloped]{$<$}(20);
  \draw[white](10)--node[black]{$<$}(20);
\end{tikzpicture}
\end{aligned}
\label{eq:threecases}
\end{equation}
The concatenation of these edges gives a path in $\Pi_s$.
Case 4 corresponds to a horizontal bottom edge in the path which is not allowed in the definition of $\Pi_s$. Therefore we have established a correspondence between the possibilities of orderings in the product and $\Pi_s$. 
\end{proof}

Denote by $A_s^*$ the set of $s$-array pairs $p$ satisfying condition $(*)$, and $\Omega_\pi := \{p \in A_s^*: \pi_p = \pi\}$. Then
\begin{align*}
  \tilde{w}_{s,N} = \left( {b - a \over N} \right)^s \sum_{\pi \in \Pi_s} \sum_{p \in \Omega_\pi} \theta_p(\pi).
\end{align*}

Given a path $\pi \in \Pi_s$, by the correspondence in \eqref{eq:threecases} there exist $m_1$, $m_2$, $\dots$, $m_{s - 1} \in \{1, 2\}$ such that for any $p \in \Omega_\pi$, $p_{x, m_x} = p_{x + 1, m_{x + 1}'}$, where $m_x' := 3 - m_x$. 
Therefore, $\Omega_\pi$ is characterised by a partial ordering on the $s + 1$ coordinates $p_{1, m_1'}, p_{1, m_1}, p_{2, m_2}, \dots, p_{s - 1, m_{s - 1}}, p_{s, m_{s - 1}'}$.
We call them the {\em essential coordinates} of $p$.
This also shows we can associate $\pi$ with $(m_1, m_2, \dots, m_{s - 1})$.
In the following we do not differentiate between $\pi$ and the corresponding partial ordering.

Any ordering $\pi \in \Pi_s$ can be decomposed into (strict) total orderings without any repetition of the essential coordinates and those with repeated essential coordinates.
We call any the former orderings $B$ a {\em linear extension} of $\pi$ which is denoted by $B \vdash \pi$, and the latter degenerate orderings, which, for reasons that will emerge in the proof of Lemma \ref{l:h} are ignored.
Thus we have 
\begin{align*}
  \Omega_\pi = \bigsqcup_{B \vdash \pi} B \cup \text{ set of degenerate orderings}.
\end{align*}

For any $p \in \Omega_\pi$, there exists a $B \vdash \pi$ such that $p \in B$.
Denote $(j_p, k_p) = (p_{1, m_1'}, p_{s, m_{s - 1}'})$.
In the total ordering imposed by $B$, let $r_B$ be the number of essential coordinates less than $j_p$ and $r_B'$ the number of those greater than $k_p$.
That is, the essential coordinates are ordered as follows,
\begin{align*}
    1 < l_1 < \dots < l_{r_B} < j_p < m_1 < \dots < m_{s - 1 - r_B - r_B'} &< k_p < n_1 < \dots < n_{r_B'} \le N, \\
    & \text{ if } r_B + r_B' < s\\
    1 < l_1 < \dots < l_{s - r_B'} < k_p < m_1 < \dots < m_{r_B + r_B' - s - 1} &< j_p < n_1 < \dots < n_{s - r_B} \le N, \\
    & \text{ if } r_B + r_B' > s
\end{align*}
We call $(r_B, r_B')$ the {\em rank} of $B$.

Let $\epsilon(\pi)$ be the number of upper vertices of the path $\pi$.
Since horizontal edges between lower vertices are not allowed, there is at least one upper vertex between two consecutive lower vertices, hence
\begin{align*}
  2 \epsilon(\pi) \ge s - 1.
\end{align*}
The location (upper or lower) of the first vertex of $\pi$, the number of upper vertices $\epsilon(\pi)$ and the parity of the length of $\pi$ together determine the number of horizontal edges in $\pi$. The cases when $\epsilon(\pi) \approx {s - 1 \over 2}$ are ``saturated'', meaning there is no horizontal edge $\pi$.
This will be later specified and exploited in the proof of Lemma \ref{l:abcd}.

The weight of $\pi$ is
\begin{align*}
  \theta_p(\pi) = (- \bar \nu)^{\epsilon(\pi)}\nu^{s - \epsilon(\pi)} \ket{j_p} \bra{k_p}.
\end{align*}
So
\begin{align*}
  \tilde w_{s,N} &= \left( \lambda {b - a \over N} \right)^s \sum_{\pi \in \Pi_s} (- \bar \nu)^{\epsilon(\pi)} \nu^{s - \epsilon(\pi)} \sum_{p \in \Omega_\pi} \ket{j_p} \bra{k_p}\\
  &\simeq \left(\lambda {b - a \over N}\right)^s \sum_{\pi \in \Pi_s} (- \bar\nu)^{\epsilon(\pi)} \nu^{s - \epsilon(\pi)} \sum_{B \vdash \Omega_\pi} \sum_{p \in B} \ket{j_p}\bra{k_p}\\
  &= \lambda^s \sum_{\pi \in \Pi_s} (- \bar \nu)^{\epsilon(\pi)} \nu^{s - \epsilon(\pi)} \sum_{B \vdash \pi} H^N_s(r_B, r_B') + v_{s, N} =: w_{s, N} + v_{s, N},
\end{align*}
where $v_{s, N}$ is the contribution from the degenerate orderings, on which one can carry out the same calculation for $w_{s, N}$ below, and that
{\footnotesize
\begin{equation*}
  H^N_s(r, r') = \begin{cases}
    \left( {b - a \over N} \right)^s \sum_{1 < l_1 < \dots < l_r < j < m_1 < \dots < m_{s - 1 - r - r'} < k < n_1 < \dots < n_{r'} \le N} \ket{j} \bra{k}, & r + r' < s\\
    \left( {b - a \over N} \right)^s \sum_{1 < l_1 < \dots < l_{s - r'} < k < m_1 < \dots < m_{r + r' - s - 1} < j < n_1 < \dots < n_{s - r} \le N} \ket{j} \bra{k}, & r + r' > s
  \end{cases}.
\end{equation*}
}

For example, for the following path $\pi$ of length $2$,
\tikzstyle{every node}=[circle, draw, fill=black!50, inner sep=0pt, minimum width=4pt]
\begin{center}
  \begin{tikzpicture}
    \draw(1,0)node{}--(2,1)node{}--(3,0)node{};
    \tikzstyle{every node}=[]
    \begin{scope}[shift={(0,-2)}]
      \foreach \i in {1,...,3}
      {
      	\node(\i1) at (\i,1) {$p_{\i,1}$};
      	\node(\i2) at (\i,0) {$p_{\i,2}$};
	\draw[white](\i1)--node[black,sloped]{$<$}(\i2);
      }
      \draw[white](11)--node[black]{$=$}(21);
      \draw[white](21)--node[black]{$<$}(31);
      \draw[white](12)--node[black]{$<$}(22);
      \draw[white](22)--node[black]{$=$}(32);
    \end{scope}
  \end{tikzpicture}
\end{center}
The ordering of the essential coordinates imposed by $\pi$ is:
\begin{align*}
  (p_{1, 1} < p_{1, 2} < p_{2, 2}) \wedge (p_{2, 1} < p_{3, 1} < p_{2, 2}),
\end{align*}
and the non-repeated starting and ending coordinates are $j_p = p_{12}$ and $k_p = p_{31}$.
The total ordering decomposition of $\Omega_\pi$ is
\begin{align*}
  \Omega_\pi = \{p_{11} < p_{12} < p_{31} < p_{22}\} \sqcup \{p_{11} < p_{31} < p_{12} < p_{22}\} \sqcup \{p_{11} < p_{12} = p_{31} < p_{22}\}.
\end{align*}
The last term is a degenerate case as $p_{12}$ is repeated. 
There is only one upper vertex, hence this path contributes $- \nu |\nu|^2 (H^N_3(2, 3) + H^N_3(3, 2))$ to $w_{s,N}$.

Define the Volterra-type kernels $>_{a}^{b}(x,y) := 1_{a \le y < x < b}$ and $<_{a}^{b}(x,y) := >_a^b(y,x)$, and $[m,n,p](x, y) := {(x - a)^m \over m!}{(y - x)^n \over n!}{(b - y)^p \over p!}$ and $[m,n,p]^\dagger(x, y) := [m,n,p](y, x)$.
The asymptotics of $H^n_s$ can be written down explicitly.

\begin{lemma}\label{l:h}
  $H^N_s(r, r')$ converges weakly to an integral operator $H_s(r, r')$, with the integral kernel
  \begin{equation*}
    h_s(r, r') = \begin{cases}
      [r, s - 1 - r - r', r'] <_a^b,& r + r' < s \\
      [s - r', r + r' - s - , s - r]^\dagger >_a^b,& r + r' > s
      \end{cases}
  \end{equation*}.
\end{lemma}

\begin{proof}
Suppose $r + r' < s$ (the case $r + r' > s$ can be done in the same way). Then
\begin{align*}
  H_s^N(r, r') &= \left( {b - a \over N} \right)^s \sum_{1 \le j < k \le N} \ket{j} \bra{k} \sum_{1\le l_1 < \dots < l_{r} < j < m_1 < \dots < m_{s - 1 - r' - r} < k < n_1 < \dots < n_{r'} \le N} 1\\
  &= \left( {b - a \over N} \right)^s \sum_{1 \le j < k \le N} \ket{j} \bra{k} {j - 1 \choose r} {k - j - 1 \choose s - 1 - r' - r} {N - k \choose r'}
\end{align*}
We denote $\Delta_N := {b - a \over N}$, then the kernel of $H_s^N(r, r')$ is
\begin{align*}
  h_s^N(x,& y) = \sum_{1 \le j < k \le N} \mathbb I_{A_j}(x) \mathbb I_{A_k}(y) {1 \over r! (s - 1 - r - r')! r'!} \\
  &\times \prod_{\alpha = 0}^{r - 1} (x_{j - 1} - a - \alpha \Delta_N) \prod_{\beta = 0}^{s - 2 - r - r'} (x_{k - 1} - x_j - \beta \Delta_N) \prod_{\gamma = 0}^{r' - 1} (b - x_k - \gamma \Delta_N).
\end{align*}
This, as $N \to \infty$, converges weakly (as an integral kernel) to $[r, s - 1 - r - r', r'] <_a^b(x, y)$.
\end{proof}

It can also be seen from the proof of this lemma that the degenerate orderings contribute $0$ to the total sum.
More specifically, the degenerate version of $h_s^N(x,y)$ where there are $d$ repeated essential coordinates is
\begin{align*}
  q_s^N&(x, y) = \sum_{1 \le j < k \le N} \mathbb I_{A_j}(x) \mathbb I_{A_k}(y) {1 \over r! (s - 1 - r - r')! r'!} {(b - a)^d \over N^d} \\
  &\times \prod_{\alpha = 0}^{r - 1} (x_{j - 1} - a - \alpha \Delta_N) \prod_{\beta = 0}^{s - d - 2 - r - r'} (x_{k - 1} - x_j - \beta \Delta_N) \prod_{\gamma = 0}^{r' - 1} (b - x_k - \gamma \Delta_N) \to 0
\end{align*}
as $N \to \infty$.
We will examine carefully the rate of convergence of this lemma and the (in)significance of the degenerate orderings later in the proof of Theorem \ref{t:conv}.
This lemma immediately gives the following corollary.

\begin{corollary}
  There exist two integer arrays $(D_{m,n,p;q})_{m, n, p\ge0, 0 \le q \le m + n + p + 1}$ and $(E_{m,n,p;q})_{m, n, p \ge0, 0 \le q \le m + n + p + 1}$ such that $w_{s, N}$ converges weakly as $N \to \infty$ to $w_s$ with kernel
  \begin{align*}
    f_s <_a^b + g_s >_a^b
  \end{align*}
  where $f_s$ and $g_s$ are defined by
  \begin{align*}
    f_s = \sum_{m, n, p\ge 0, m + n + p = s-1} \sum_{q = 0}^{s} D_{m, n, p; q} (-\bar \nu)^q \nu^{s - q} [m, n, p],\\
    g_s = \sum_{m, n, p\ge 0, m + n + p = s-1} \sum_{q = 0}^{s} E_{m, n, p; q} (-\bar \nu)^q \nu^{s - q} [m, n, p]^\dagger.
  \end{align*}
\end{corollary}

Indeed, $D_{m, n, p; q}$ (resp. $E_{m, n, p; q}$) enumerates the linear extensions of all possible paths of length $m + n + p$ with $q$ upper vertices and rank $(m, p)$ (resp. $(m + n + 1, n + p + 1)$).

\begin{corollary}\label{c:symmetry}
  The functions $f_s$ and $g_s$ both are symmetric in the following sense:
  \begin{align*}
    f_s(x, y) = f_s(a + b - y, a + b - x),\qquad g_s(x, y) = g_s(a + b - y, a + b - x)
  \end{align*}
\end{corollary}
\begin{proof}
  This follows from the fact that the path inversion $(i, b_i) \mapsto (i, b_{s + 1 - i})$ is a weight-preserving bijection between $\Pi_s$ and itself.
\end{proof}

For example, some calculation yields
\begin{equation}
\begin{aligned}
f_1 &= - \bar \nu [0,0,0],\\
f_2 &= -|\nu|^2 [0,0,1] - |\nu|^2 [1,0,0] + {\bar\nu}^2[0,1,0],\\
f_3 &= (\bar\nu |\nu|^2 - {\bar\nu}^3)[0,2,0] + \bar\nu|\nu|^2[0,1,1] + \bar\nu|\nu|^2[1,1,0] - \nu |\nu|^2[1,0,1],\\
g_1 &= - \bar \nu [0, 0, 0]^\dagger,\\
g_2 &= 0,\\
g_3 &= - \nu |\nu|^2 [1, 0, 1]^\dagger.
\end{aligned}
  \label{eq:123}
\end{equation}

The following three theorems are the main results of this paper:

\begin{theorem}
  The closed form expression of $D$ and $E$ are:
      \begin{equation}
      D_{m, n, p; q} =
      \begin{cases}
        {n \choose q - 1} - {n \choose q}, & 2 q > m + n + p\\
        {n \choose q - m} - {n \choose q}, & 2 q = m + n + p\\
        0, & 2 q < m + n + p
      \end{cases}
      \label{eq:D1}
    \end{equation}
    \begin{align*}
      E_{m, n, p; q}&= \mathbb I_{m = p = q,n=0}
    \end{align*}
    \label{t:main}
\end{theorem}
\begin{proof}
  See Section \ref{s:pfmain}.
\end{proof}

\begin{theorem}
  The operator $W_N$ converges weakly to
  \begin{align*}
    W = I + \sum_{s \ge 1} w_s
  \end{align*}
  \label{t:conv}
\end{theorem}
\begin{proof}
  See Section \ref{s:pfconv}.
\end{proof}

For $j \ge 0$, let $B_j$ be power series in two variables related to the Bessel functions of the first kind $J_j$.
\begin{align*}
B_j(x, y) := \sum_{n \ge 0}{(-1)^{n + j} x^{n + j} y^n \over (n + j)! n!} = (-1)^j (x / y)^{j / 2} J_j(2 \sqrt{x y}).
\end{align*}
Let $I$ be the identity, then the kernel of the operator $W - I$ can be written in terms of $B_j$.

\begin{theorem}
  The integral operator $W - I$ has kernel
  \begin{align*}
    \text{ker} (W - I) (x, y) &= \Bigg(\nu B_0( (y - a) |\nu|, (b - x) |\nu|) + |\nu| B_1( (b - a) |\nu|, (y - x) |\nu|)\\
    &- (\nu + \bar \nu) \sum_{q \ge 0} B_q( (y - x) |\nu|, (b - a) |\nu|) \left({\bar\nu \over |\nu|}\right)^q\Bigg) <_a^b(x, y) \\
    &+ \nu B_0( (y - a) |\nu|, (b- x) |\nu|) >_a^b(x, y).
  \end{align*}
  Moreover, $W$ is unitary.
  \label{t:unitarity}
\end{theorem}
\begin{proof}
  See section \ref{s:unitarity}.
\end{proof}
For example, when $\mu = 0$ and $\lambda > 0$, the kernel of the operator corresponding to the L\'evy stochastic area is
\begin{align*}
    \text{ker} &(W - I) (x, y) = \Bigg(\lambda B_0( (y - a) \lambda, (b - x) \lambda) + \lambda B_1( (b - a) \lambda, (y - x) \lambda)\\
    &- 2 \lambda \sum_{q \ge 0} B_q( (y - x) \lambda, (b - a) \lambda) \Bigg) <_a^b(x, y) + \lambda B_0( (y - a) \lambda, (b- x) \lambda) >_a^b(x, y).
\end{align*}
  Moreover by plugging $D_{m, n, p; q}$ and $E_{m, n, p; q}$ into the integral identity \eqref{eq:int} below, the unitarity of $W$ implies the following combinatorial identity:
  \begin{align*}
    D_{\alpha, \beta, \gamma; \xi} &- \mathbb I_{\alpha = \gamma = \xi - 1, \beta = 0} - {\alpha + \gamma - 1 \choose \gamma} \mathbb I_{\beta + \gamma + 1 = \alpha = \xi} \\
    &- \sum_{m = 0}^\alpha \sum_{p = 0}^{\gamma - \alpha + m} \sum_{n = 0}^{\alpha + \beta -\gamma - m + p -1} D_{m, n, \alpha + \beta - \gamma - m - n + 2 p - 1; \xi - \gamma + p - 1} \\
    &\qquad\qquad\times{\alpha \choose m} {\gamma - \alpha + m + n - p \choose n} {\gamma \choose p} \\
    &+ \sum_{m_1 = 0}^\alpha \sum_{m_2 = 0}^\beta \sum_{n_1 = 0}^{\gamma - 1} \sum_{n_2 = 0}^{\gamma - 1 - n_1} \sum_{p_1 = 0}^{\gamma - 1 - n_1 - n_2} \sum_{t_1 = 0}^\xi (-1)^{\alpha + \gamma - m_1 - n_1 - p_1 + m_2} \\
    &\qquad D_{m_1, n_1 + \beta - m_2, p_1; t_1} D_{m_2 + \alpha - m_1, n_2, \gamma - 1 - n_1 - n_2 - p_1; \alpha + \gamma -\xi - m_1 - n_1 - p_1 + m_2 + t_1} \\
    &\qquad\qquad \times {\alpha \choose m_1} {\beta \choose m_2} {n_1 + n_2 \choose n_1} {\gamma - 1 - n_1 - n_2 \choose p_1} = 0.
  \end{align*}

\section{Dyck paths and Catalan numbers}\label{s:pfmain}
For $m \in \mathbb Z_{\ge 0}$ and $n \in \mathbb Z$, define the binomial coefficient the usual way
\begin{align*}
  {m \choose n} := {m! \over n! (m - n)!} \mathbb I_{0 \le n \le m}.
\end{align*}
For integers $m,n,p$ define a double generalisation of the Catalan numbers and the Catalan's triangle
  \begin{align*}
    C_{m,n,p}:={m+n\choose m}-{m+n\choose m+p+1}.
  \end{align*}
For $\alpha, m, n, p \in \mathbb Z_{\ge 0}$, denote by $T_{\alpha, m, n, p}$ the set of lattice paths $(\rho_i)_{i = 0}^{m + n}$ such that $\rho_0 = \alpha$, $|\rho_i - \rho_{i - 1}| = 1$, $\rho_i \ge - p$, $\rho_{m + n} = \alpha + m - n$.
That is, $T_{\alpha, m, n, p}$ is the set of Dyck paths starting from $\alpha$, having $m$ up-steps, $n$ down-steps that never cross the line $y = - p$.
By the reflection principle we obtain the following lemma, which shows these numbers have a similar combinatorial interpretation to the Catalan numbers.
\begin{lemma}
  When $m, n, p\ge0$ and $m - n \ge - p - 1$, $C_{m, n, p} = |T_{0, m, n, p}|$.
\end{lemma}
The doubly generalised Catalan numbers have been discussed in e.g. \cite{reuveni14}.

When $m,n\ge0$ and $p=0$, $C_{m,n,0}$ is reduced to the $(m,n)$th entry in the Catalan triangle (OEIS:A009766) which we denote by $C_{m,n}$; furthermore when $m = n$, $C_{n, n}$ is the $n$th Catalan number which we denote by $C_n$. 

The following recurrence relation will be useful:
\begin{lemma}
If $n \ge 0$, $m \ge p$ and $m + n + p + 1 \ge 0$, then
\begin{align*}
  \sum_{k = 0}^{\lfloor {m + p \over 2} \rfloor} C_{k + n, k} C_{m - k, p - k} = C_{m + n + 1, p}
\end{align*}
  \label{l:recurrence}
\end{lemma}
\begin{proof}
  We first show a basic version of this formula is true: for $n \ge 0$, $m \ge p \ge 0$,
  \begin{align*}
    \sum_{k = 0}^pC_{k + n, k}C_{m - k, p - k} = C_{m + n + 1, p}.
  \end{align*}
  This can be proved using a combinatorial argument similar to one used to prove the recurrence relation of the Catalan numbers which is a special case of the identity above:
  \begin{align*}
    \sum_{k = 0}^p C_k C_{p - k} = C_{p + 1}.
  \end{align*}
  Define a ``stopping time'' $\sigma$ on $T_{0, m + n + 1, p, 0}$ by
  \begin{align*}
    \sigma(\rho) = \max\{i \ge 0: \rho_i = n\}.
  \end{align*}
  Then
  \begin{align*}
    C_{m + n + 1, p} = &\sum_{k = 0}^p |\{ \rho \in T_{0, m + n + 1, p, 0}: \sigma(\rho) = 2 k + n\}| \\
    &= \sum_{k = 0}^p |T_{0, k + n, k, 0}| |T_{n + 1, m - k, p - k, n + 1}| \\
    &= \sum_{k = 0}^p |T_{0, k + n, k, 0}| |T_{0, m - k, p - k, 0}| = \sum_{k = 0}^p C_{n + k, k} C_{m - k, p - k}.
  \end{align*}

If the condition $n \ge 0, m \ge p$ are retained, but $p < 0$ and $m + n + 1 + p \ge 0$, then the LHS is zero because the domain of the summation is empty. 
The RHS is also zero because ${m + n + 1 + p \choose m + n + 1} = {m + n + 1 + p \choose m + n + 2} = 0$.

Since $m \ge p$, we have $p \le \lfloor{m + p \over 2}\rfloor$. 
Moreover, for any $k \in (p, \lfloor{m + p \over 2}\rfloor]$, $C_{m - k, p - k} = {m + p - 2 k \choose m - k} - {m + p - 2 k \choose m - k + 1} = 0$.
Therefore we can extend the domain of the summation from $0 \le k \le p$ to $0 \le k \le \lfloor{m + n \over 2}\rfloor$.
\end{proof}

\begin{lemma}
  For any $B \vdash \pi \in \Pi_s$, if $\pi(0) = 1$ then $r_B = 0$, and if $\pi(0) = 0$ then $r_B > 0$. If $\pi(s) = 1$ then $r_B' = 0$, and if $\pi(s) = 0$ then $r_B' > 0$.
  \label{l:01}
\end{lemma}
\begin{proof}
  We show the claim for $\pi(0)$, as the one for $\pi(s)$ can be deduced from the symmetry property.
  If $\pi(0) = 1$, then $\pi(1) = 0$ or $1$.
  If $\pi(1) = 0$ then by the correspondence \eqref{eq:threecases}, for any $p \in \Omega_\pi$, the first four coordinates have the ordering $j_p = p_{1,1} < p_{2, 1} < p_{1, 2} = p_{2, 2}$. Since $p_{1, 1} \le p_{k, 1}, k \ge 2$, and $p_{1, 1} < p_{1, 2} \le p_{k, 2}, k \ge 2$. Thus $j_p$ is the smallest (essential) coordinate and $r_B = 0$.
  If $\pi(1) = 1$ then $j_p = p_{1, 1} < p_{1, 2} = p_{2, 1} < p_{2, 2}$ hence it's also the smallest coordinates and $r_B = 0$.

  If $\pi(1) = 0$, then $\pi(1) = 1$ and by \eqref{eq:threecases}, for any $p \in \Omega_\pi$, the first four coordinates are ordered as $p_{1, 1} = p_{2, 1} < p_{1, 2} = j_p < p_{2, 2}$.
  Hence $j_p$ is greater than at least one other essential coordinate and $r_B > 0$.
\end{proof}

In some extreme cases the coefficient $D_{m, n, p; q}$ can be calculated directly.
We denote by $D^\pi_{m, n, p; q}$ the contribution to $D_{m, n, p; q}$ from path $\pi$.
\begin{lemma}
\begin{itemize}
  \item (Case A) $D_{0, 2 k, 0; k + 1} D^{\vee^k}_{0, 2 k, 0; k + 1} = C_k$. Conversely, if $m = 0, 2 q = n + p + 2$, then $D_{m, n, p; q} > 0$ only if $p = 0, 2 q - 2 = n$.
  \item (Case B) $D_{0, n, 2 k - n + 1; k + 1} = D^{\setminus \wedge^k}_{0, n, 2 k - n + 1, k + 1} = C_{k, n - k}$ for $0 \le n \le 2 k + 1$.
  \item (Case C) $D_{2 k - n + 1, n, 0; k + 1} = D^{\wedge^k /}_{2 k - n + 1, n, k + 1} = C_{k, n - k}$ for $0 \le n \le 2 k + 1$.
  \item (Case D) $D_{r, 2 k - r - r', r'; k} = D^{\wedge^k}_{r, 2 k - r - r', r'; k} = C_{k - r, k - r', r - 1} = C_{k - r', k - r, r' - 1}$ for $0 \le r + r' \le 2 k$.
\end{itemize}
Moreover, $E_{m, n, p; q} = \ind_{m = p = q, n = 0}$.
  \label{l:abcd}
\end{lemma}
\begin{proof}
  First we show the first identity in each case.
  In Case D, for there to be $k$ upper vertices and $k + 1$ lower vertices, the path can only be $\wedge^k$.

  For Case A, since the rank is $(0, 0)$, by Lemma \ref{l:01}, any path $\pi$ contributing to $D_{0, 2 k, 0; k + 1}$ has to begin and end with upper vertices.
  Removing these two vertices resulting a path of length $2 k - 2$, $k - 1$ upper vertices and $k$ lower vertices, which is the same as Case D.
  Therefore $\pi = \vee^k$.

  With the same arguments the paths for Case B and C are also determined to be $\setminus\wedge^k$ and $\wedge^k/$ respectively.

  Now we show the second identity in Case D, as Cases A, B and C are simpler variations of D and can be verified similarly. We achieve this by associating partial orderings with sets of Dyck paths.
  The path $\pi = \wedge^k$ imposes the following ordering of the essential coordinates:
  \begin{align*}
    \begin{array}{ccccccccccc}
      p_{1, 1} & < & p_{3, 1} & < & p_{5, 1} & < & \dots & < & p_{2 k - 1, 1} & < & k_p = p_{2 k + 1, 1}\\
      \wedge &   & \wedge &   & \wedge &   & \dots & & \wedge & & \wedge \\
      j_p = p_{1, 2} & < & p_{2, 2} & < & p_{4, 2} & < & \dots & < & p_{2 k - 2, 2} & < & p_{2 k, 2}
    \end{array}
  \end{align*}
  We relabel these coordinates by $t_{1, 1} := p_{1, 1}, t_{2, 1} = p_{3, 1}$ and so on, to obtain
  \begin{equation}
    \begin{array}{ccccccccccc}
      t_{1, 1} & < & t_{2, 1} & < & t_{3, 1} & < & \dots & < & t_{k, 1} & < & t_{k + 1, 1}\\
      \wedge &   & \wedge &   & \wedge &   & \dots & & \wedge & & \wedge \\
      t_{1, 2} & < & t_{2, 2} & < & t_{3, 2} & < & \dots & < & t_{k, 2} & < & t_{k + 1, 2}
    \end{array}
    \label{eq:array}
  \end{equation}
  There is a one-one correspondence between the set of all linear extensions of this partial ordering (namely $\{B: B \vdash \pi\}$) and $T_{0, k + 1, k + 1, 0}$.
  The Dyck path $\rho$ corresponding to the linear extension $t_{m_1, b_1} < t_{m_2, b_2} < \dots < t_{m_{k + 1}, b_{k + 1}}$ is defined by
  \begin{align*}
    \rho(i) =
    \begin{cases}
      \rho(i - 1) + 1, & \text{ if } b_i = 1\\
      \rho(i - 1) - 1, & \text{ if } b_i = 2
    \end{cases}
  \end{align*}
  Clearly, the rank of a linear extension $B$ being $(r, r')$ is equivalent to the corresponding Dyck path starting with $r$ up-steps followed by a down-step and concluding with one down-step with $r'$ up-steps.
  These cut off the first $r + 1$ and the last $r' + 1$ steps corresponds to $T_{r - 1, k - r, k - r', 0}$.
  Therefore 
  \begin{align*}
  D^{\wedge^k}_{r, 2 k - r - r', r'; k} = |T_{r - 1, k - r, k - r', 0}| = |T_{0, k - r, k - r', r - 1}| = C_{k - r, k - r', r - 1}.
  \end{align*}
  If $r = 0$, $r' = 0$, then by Lemma \ref{l:01} $D^{\wedge^k}_{r, 2 k - r - r', r'; q} = 0$, which agrees with $C_{k - r, k - r', r - 1}$ as well.
  On the other hand, since the paths of $T_{0, k, k, 0}$ only have $k$ up- and down-steps, the LHS is 0 if $r > k$ or $r' > k$, which agrees with the right hand side.

  Finally, the $\wedge^k$ are the only possible paths to contribute to the coefficients $E$, which record the instances when $k_p < j_p$.
  The corresponding linear extension $B \vdash \wedge^k$ is $\{t_{1,1} < t_{2,1} < \dots < t_{k + 1, 1} < t_{1, 2} < t_{2, 2} < \dots < t_{k + 1, 2}\}$.
  For any other paths, by Lemma \ref{l:01}, any path starting with $\setminus$ or ending with $/$ has $j_p$ as the smallest or $k_p$ as the greatest essential coordinate; on the other hand, any horizontal edge will result in $j_p < k_p$.
\end{proof}

From the above proof we can deduce a stronger version of Lemma \ref{l:01}: $D^\pi_{0, n, p; q} \neq 0$ only if $\pi(0) = 1$, and for $m \ge 1$, $D^\pi_{m, n, p; q} \neq 0$ only if $\pi(0) = 0$ and begins with $\wedge^{m - 1}$. We also refer to this stronger version as Lemma \ref{l:01}.

\begin{proof}[Proof of Theorem \ref{t:main}]
Case D covers the $2 q = m + n + p$ case; 
moreover, by the same argument as in the proof of the first identities in each case of Lemma \ref{l:abcd}, for $\pi \in \Pi_{m + n + p + s}$,
\begin{align*}
  \epsilon(\pi) + 1 \ge m + n + p + 1 - \epsilon(\pi)
\end{align*}
therefore $D_{m, n, p; q} \neq 0$ only if $2 q \ge m + n + p$.
Thus it suffices to show
\begin{align*}
D_{m, n, p; q} = C_{q - 1, n - q + 1},\quad 2 q > m + n + p.
\end{align*}

We group the paths into ones starting with $\wedge^k/-$ (call the set of such paths $\Pi^{\wedge k}$) and ones starting with $\vee^k-$ (call the set of such paths $\Pi^{\vee k}$).
Then by Lemma \ref{l:01} $D_{m, n, p; q}$ are contributed from $\Pi^{\wedge k}$ if $k - 1 \ge m > 0$, and from $\Pi^{\vee^k}$ if $m = 0$:
\begin{align*}
D_{m, n, p; q} = 
\begin{cases}
\sum_{k \ge m - 1} \sum_{\pi \in \Pi^{\wedge k}} D^\pi_{m, n, p; q}, & m > 0\\
\sum_{k \ge 0} \sum_{\pi \in \Pi^{\vee k}} D^\pi_{m, n, p; q}, & m = 0
\end{cases}
\end{align*}
where $D^\pi_{m, n, p; q}$ is the contribution from path $\pi$ to the coefficient $D_{m, n, p; q}$.
Therefore we divide the proof into two cases, $m = 0$ and $m > 0$.
The formula for $D_{m, n, p; q}$ with $m + n + p \le 2$ can be verified by hand (the reader can check their calculation against \eqref{eq:123}), so we assume the formula is true for $s \le m + n + p$, and we want to use induction to verify the formula of $D_{m, n, p; q}$ in general.

\subsection{$m > 0$}
When $m > 0$, for any $\pi \in \Pi^{\wedge k}$, 
\begin{align*}
D^\pi_{m, n, p; q} = D_{m, 2 k + 1 - m, 0; k + 1} D^{\theta_{2 k + 2} \pi}_{0, m + n - 2 k - 2, p; q - k - 1},
\end{align*}
where $\theta_r: \Pi_s \to \Pi_{s - r} \forall s \ge r$ is the shifting operator such that $(\theta_r \pi)(j) = \pi(j + r)$.

Summing over $k$ and $\pi \in \Pi^{\wedge k}$ we have
\begin{align*}
D_{m, n, p; q} &= \sum_{k \ge m - 1} \sum_{\pi \in \Pi^{\wedge k}} D^\pi_{m, n, p; q} \\
&= \sum_{k = m - 1}^{\lfloor{m + n - 2 \over 2}\rfloor} \sum_{\pi \in \Pi^{\wedge k}} D_{m, 2 k + 1 - m, 0; k + 1} D^{\theta_{2 k + 2} \pi}_{0, m + n - 2 k - 2, p; q - k -1} \\
&= \sum_{k = m - 1}^{\lfloor{m + n - 2 \over 2}\rfloor} D_{m, 2 k + 1 - m, 0; k + 1} D_{0, m + n - 2 k - 2, p; q - k - 1} \\
&= \sum_{k = m - 1}^{\lfloor{m + n - 2 \over 2}\rfloor} C_{k, k - m + 1} C_{q - k - 2, m + n - q - k} \\
&= \sum_{k = 0}^{\lfloor{n - m \over 2}\rfloor} C_{k + m - 1, k} C_{q - m - 1 - k, n - q + 1 - k}
\end{align*}
where the last two equalities comes from the Case C and the induction assumption.

To apply Lemma \ref{l:recurrence}, we check the three conditions hold:
(1) The condition ``$n \ge 0$'' becomes $m - 1 \ge 0$: this is correct as $m > 0$.
(2) ``$m \ge p$'' is $2 q \ge m + n + 2$: we know that $2 q > m + n + p$, so either $2 q \ge m + n + 2$ or $2 q = m + n + 1$, in the latter since $m + n + 1 \le m + n + p + 1 \le 2 q$, we have $p = 0$ and this is covered by Case C.
(3) ``$m + n + p + 1 \ge 0$'' becomes $n \ge 0$, which is evidently true by the definition of $D_{m, n, p; q}$.
The upper bound of the summation domain ``$\lfloor{m + p \over 2}\rfloor$'' becomes $\lfloor{n - m \over 2}\rfloor$. 
Therefore we can apply Lemma \ref{l:recurrence} to the sum above and obtain
\begin{align*}
D_{m, n, p; q} = C_{q - 1, n - q + 1}.
\end{align*}

\subsection{$m = 0$}
When $m = 0$, similarly, for any $\pi \in \Pi^{\vee k}$,
\begin{align*}
D^\pi_{0, n, p; q} = D_{0, 2 k, 0; k + 1} D^{\theta_{2 k + 1} \pi}_{0, n - 2 k - 1, p; q - k - 1}.
\end{align*}
Again, summing over $k$ and $\pi \in \Pi^{\vee k}$ we have
\begin{align*}
D_{0, n, p; q} = \sum_{k \ge 0} \sum_{\pi \in \Pi^{\vee k}} D^\pi_{m, n, p; q} = \sum_{k = 0}^{\lfloor{n - 1 \over 2}\rfloor} \sum_{\pi \in \Pi^{\vee k}} D_{0, 2 k, 0; k + 1} D^{\theta_{2 k + 1} \pi}_{0, n - 2 k - 1, p; q - k - 1} \\
= \sum_{k = 0}^{\lfloor{n - 1 \over 2}\rfloor} D_{0, 2 k, 0; k + 1} D_{0, n - 2 k - 1, p; q - k - 1} = \sum_{k = 0}^{\lfloor{n - 1 \over 2}\rfloor} C_{k, k} C_{q - k - 2, n - q - k + 1}
\end{align*}
where the last equality comes from Case A, the induction assumption and the fact that $C_k = C_{k, k}$.
Once again, we want to check the conditions in order to apply Lemma \ref{l:recurrence}.
The condition ``$n \ge 0$'' is obvious.
``$ m \ge p $'' is equivalent to $2 q \ge n + 3$. 
Since $2 q > n + p$, there are two possibilities apart from ``$m \ge p$'':
\begin{enumerate}
  \item $2 q = n + p + 1$. This is covered by Case B.
  \item $2 q = n + p + 2$ and $p = 0$. This is covered by Case A.
\end{enumerate}
``$m + n + p + 1 \ge 0$'' is again equivalent to $n \ge 0$, which is evidently true.
The upper bound of the summation domain ``$\lfloor{m + p \over 2}\rfloor$'' is $\lfloor{n - 1 \over 2}\rfloor$. 
Therefore we can apply Lemma \ref{l:recurrence} to the sum above and obtain:
\begin{align*}
D_{0, n, p; q} = C_{q - 1, n - q + 1}.
\end{align*}
\end{proof}

\section{Proof of Theorem \ref{t:conv}}\label{s:pfconv}
In this section we often abuse notations and do not differentiate between operators and their kernels.
Without loss of generality assume $|\nu| = 1$ (otherwise one can scale $(a, x, y, b)$).
We only consider the generating function of coefficient $D$, as the case for $E$ can be dealt with similarly.
We write 
\begin{align*}
  \phi^N(m, n, p) &= \sum_{1 \le j < k \le N} \ind_{A_j}(x) \ind_{A_k}(y) \prod_{\alpha = 0}^{m - 1} (x_{j - 1} - a - \alpha \Delta_N) \prod_{\beta = 0}^{n - 1} (x_{k - 1} - x_j - \beta \Delta_N)\\
  &\qquad\times\prod_{\gamma = 0}^{p - 1} (b - x_k - \gamma \Delta_N)\\
  \phi(m, n, p) &= (x - a)^m (y - x)^n (b - y)^p \ind_{a \le x < y < b},
\end{align*}
where we recall $\Delta_N = {b - a \over N}$.
Then $0 \le \phi^N(m, n, p) \le \phi(m, n, p) \le (b - a)^{m + n + p}$.
We also write
\begin{align*}
  \D_{m, n, p} = \sum_{q = 0}^s {D_{m, n, p; q} \over m! n! p!} (-1)^q \nu^{m + n + p + 1 - 2 q}.
\end{align*}
Then
\begin{align*}
  W_N = \sum_{s = 1}^{N - 1} w_{s, N} + \sum_{s = 1}^{2 N - 3} v_{s, N}
\end{align*}
where $v_{s, N}$ are the degenerate terms, and
\begin{align*}
  w_{s, N} = \sum_{m + n + p = s - 1} \D_{m, n, p} \phi^N(m, n, p).
\end{align*}
The reason for the range of the sum for $s$ to be $1 \le s \le N - 1$ is because $w_{s, N} = 0$ for $s \ge N$, as $\phi^N(m, n, p) = 0$ for $m + n + p \ge N - 1$.

\begin{proof}[Proof of Theorem \ref{t:conv}]
We divide the proof into three parts:
\begin{enumerate}
  \item $\sum_{m + n + p \ge N} \D_{m, n, p} \phi^N(m, n, p) \overset{N \to \infty}{\to} 0$ uniformly on $[a, b)^2$. This shows the limit exists.
  \item $\sum_{s = 1}^{N - 1} w_{s, N}$ converges weakly to $W$.
  \item The degenerate terms vanish uniformly: $\sum_{s = 1}^{2 N - 3} v_{s, N}$ are arbitrarily small as $N$ grows bigger.
\end{enumerate}

\subsection{Part 1}
By the formula of $D_{m, n, p; q}$ a bound can be immediately obtained:
\begin{align*}
  0 \le D_{m, n, p; q} \le {n \choose \lfloor n / 2 \rfloor} \le 2^n.
\end{align*}
Similarly one can bound the trinomial coefficient ${(m + n + p)! \over m! n! p!} \le 3^{m + n + p}$. Combining these two bounds we obtain
\begin{align*}
  |\D_{m, n, p}| = \left|\sum_{q = 0}^{m + n + p + 1} (-1)^q \nu^{m + n + p + 1 - 2 q} {D_{m, n, p} \over m! n! p!}\right|
  \le (m + n + p + 1) {6^{m + n + p} \over (m + n + p)!}
\end{align*}

Therefore
\begin{align*}
  \left| \sum_{m + n + p \ge N} \D_{m, n, p} \phi^N(m, n, p)\right| &\le \sum_{m + n + p \ge N}(m + n + p + 1) {(6 (b - a))^{m + n + p} \over (m + n + p)!} \\
  &= \sum_{r \ge N} {r + 2 \choose 2} (r + 1) {(6 (b - a))^{r} \over r!} \to 0
\end{align*}
as $N \to \infty$.

\subsection{Part 2}
We want to show that for any $\epsilon > 0$ and sufficiently large $N$, we have $\langle f, \sum_{s \le N - 1} (w_{s, N} - w_s) g \rangle < \epsilon \|f\|_2 \|g\|_2$ for testing functions $f, g \in L^2([a, b))$.

Equivalently, we must show that
\begin{align*}
  \Bigg|\int_a^b \int_a^b \bar f(x, y) \sum_{m + n + p \le N - 2} \D_{m, n, p} (\phi^N(m, n, p) - \phi(m, n, p)) g(x, y) dx dy \Bigg| \le \epsilon \|f\|_2 \|g\|_2.
\end{align*}
We divide it into two further parts.
\begin{enumerate}
  \item $(x - a)^m (y - x)^n (b - y)^p \mathbb \sum_{1 \le j < k \le N} I_{A_j \times A_k}(x, y) \approx (x - a)^m (y - x)^n (b - y)^p \mathbb I_{a \le x < y < b}$,
  \item $(x - a)^m (y - x)^n (b - y)^p \sum_{1 \le j < k \le N} \mathbb I_{A_j \times A_k}(x, y) \approx  \sum_{1 \le j < k \le N} \tau^N(j, k; m, n, p) \cdot \mathbb I_{A_j}(x) \mathbb I_{A_k}(y) $;
\end{enumerate}
where
\begin{align*}
  \tau^N(j,k;m,n,p) := \prod_{\alpha = 0}^{m - 1} (x_{j - 1} - a - \alpha \Delta_N) \prod_{\beta = 0}^{n - 1} (x_{k - 1} - x_j - \beta \Delta_N) \prod_{\gamma = 0}^{p - 1} (b - x_k - \gamma \Delta_N).
\end{align*}
\subsubsection{Part 2.1}
We want to show that
\begin{align*}
  \Bigg|\int_a^b \int_a^b \bar f(x, y) \sum_{m + n + p \le N - 2} \D_{m, n, p} &\sum_{1 \le j < k \le N} (x - a)^m (y - x)^n (b - y)^p \\
  &(\mathbb I_{A_j}(x) \mathbb I_{A_k}(y) - \ind_{a \le x < y < b}) g(x, y) dx dy\Bigg| \le \epsilon.
\end{align*}
Denote the left hand side by $B_N$, then
by (1) and that $(x - a)^m (y - x)^n (b - y)^p \le (b - a)^{m + n + p}$ for $a \le x \le y \le b$,
we have that
\begin{align*}
 B_N \le \sum_{m + n + p \le N - 2} (m + n + p + 1) {(6(b - a))^{m + n + p} \over (m + n + p)!} \left|\sum_{j = 1}^N \int \int_{x_{j - 1} \le x < y < x_j} \bar f(x) g(y) dx dy\right| \\
 \le \sum_{r \le N - 1} (r + 1) {r + 2 \choose 2} {(6 (b - a))^r \over r!} \left|\sum_{j = 1}^N \int \int_{x_{j - 1} \le x < y < x_j} \bar f(x) g(y) dx dy\right|
\end{align*}
The term in the modulus can be bounded by repeated use of Cauchy-Schwartz inequality:
\begin{align*}
 \left|\sum_{j = 1}^N \int \int_{x_{j - 1} \le x < y < x_j} \bar f(x) g(y) dx dy\right| \le \sum_{j = 1}^N \int_{x_{j-1}}^{x_j} |f(x)| dx \int_x^{x_j} |g(y)| dy \\
 \le \sum_{j = 1}^N \int_{x_{j-1}}^{x_j} |f(x)| \sqrt{x_j - x} \|g\|_2 dx \le \sum_{j = 1}^N {(b - a)^2 \over 2 N^2} \|f\|_2 \|g\|_2 = {(b - a)^2 \over 2 N} \|f\|_2 \|g\|_2.
\end{align*}
Therefore
\begin{align*}
  B_N \le \sum_{r = 0}^{N - 1} (r + 1){r + 2 \choose 2} {(6 (b - a))^r \over r!} {(b - a)^2 \over 2 N} \|f\|_2 \|g\|_2 \le C N^{- 1} \|f\|_2 \|g\|_2
\end{align*}
for some constant $C$,
where the second bound comes from the fact that $\sum_{r \ge 0} (r + 1) {r + 2 \choose 2} {(6 (b - a))^r \over r!} < \infty$.
\subsubsection{Part 2.2}
We establish the following uniform convergence, from which weak convergence will follow:
{\small
\begin{align*}
  \sum_{m + n + p \le N - 2} & \D_{m, n, p} \sum_{1 \le j < k \le N} \mathbb I_{A_j}(x) \mathbb I_{A_k}(y) \left(\tau^N(j, k; m, n, p) - (x - a)^m (y - x)^n (b - y)^p \right) \to 0
\end{align*}
}
When $a \le x \le y \le b$, $\tau^N$ is non-negative, and for $m + n + p \le N$ we use a telescoping series:
\begin{align*}
  &(x - a)^m (y - x)^n (b - y)^p - \tau^N(j, k; m, n, p)\\
  &= (x - x_{j - 1}) \dots \\
   &+ (x_{j - 1} - a) (x - x_{j- 2}) \dots \\
   &+ (x_{j - 1} - a) (x_{j - 2} - a) (x - x_{j - 3}) \dots\\
   &+ \dots \\
   &+ (x_{j - 1} - a) (x_{j - 2} - a) (x_{j - 3} - a) \dots (x_{k + p - 2} - y) (b - y)\\
   &+ (x_{j - 1} - a) (x_{j - 2} - a) (x_{j - 3} - a) \dots (b - x_{k + p - 2}) (x_{k + p - 1} - y)\\
   &+ (x_{j - 1} - a) (x_{j - 2} - a) (x_{j - 3} - a) \dots (b - x_{k + p - 2}) (b - x_{k + p - 1})\\
   &- (x_{j - 1} - a) (x_{j - 2} - a) (x_{j - 3} - a) \dots (b - x_{k + p - 2}) (b - x_{k + p - 1})\\
   &\le N^{-1} (b - a)^{m + n + p} \left(\sum_{\alpha = 1}^m \alpha + \sum_{\beta = 1}^n \beta + \sum_{\gamma = 1}^p \gamma\right) \\
   &\le N^{-1} (b - a)^{m + n + p} {3 \over 2} (m + n + p) (m + n + p + 1).
\end{align*}
Therefore
\begin{align*}
  \Bigg|&\sum_{m + n + p \le N} \D_{m, n, p} \sum_{1 \le j < k \le N} \mathbb I_{A_j}(x) \mathbb I_{A_k}(y) \big(\tau^N(j, k; m, n, p) \\
  &- (x - a)^m (y - x)^n (b - y)^p \big)\Bigg| 
\le N^{-1} \sum_{r \le N}  {k + 2 \choose 2} {3 \over 2} r (r + 1)^2 {(6(b - a))^r \over r!} \le C N^{-1}.
\end{align*}

\subsection{Part 3}
The degenerate terms are the total orderings of path of length $s$ with some repeated essential coordinates.
If such a total ordering $J$ has $s + 1 - d$ non-repeated coordinates, then we call $d$ the {\em degree} of degeneration, or we say that there are $d$ degenerations in $J$.
Each degeneration happens on a wedge part of a path, that is, any two essential coordinates $p_{i_1, j_1} = p_{i_2, j_2}$ if and only if they correspond to parts of the same $\wedge^k$ part of a path for some $k$.
On the other hand, degenerations happen in pairs. That is, for an array $p$, there do not exist three essential coordinates equal to each other, which would violate the partial ordering.
Therefore, given a path of length $s - 1$, the number of degenerate total orderings with $d$ degenerations is bounded by ${{s - 1 \over 2} \choose d}^2$.
Moreover, the number of paths of length $s - 1$ is the Fibonacci number ${\Phi^{s + 2} - (- \Phi)^{- s - 2} \over \sqrt 5}$ where $\Phi = {\sqrt 5 + 1 \over 2}$.
Since a path of length $s - 1$ can have at most ${s - 1 \over 2}$ wedges, there are at most ${s - 1 \over 2}$ degenerations.

Therefore for each $N$, we have
\begin{align*}
  V_N = \sum_{s = 1}^{2 N - 3} v_{s, N}
\end{align*}
where $V_{s, N}$ is the counterpart of $W_{s, N}$ that collects all degenerate cases of paths of lengths $s - 1$.
By applying the calculation in the proof of Lemma \ref{l:h}, we can see the degenerate (of degree $d$) version of $h_s^N$ is
\begin{align*}
  q_s^N(x, y) = {(b - a)^d \over m! n! p! N^d} \phi^N(m, n, p)
\end{align*}
where $m + n + p + d = s - 1$.
Thus the sum is bounded uniformly by
\begin{align*}
  {(b - a)^{s - 1} \over m! n! p! N^d}.
\end{align*}
Therefore the total sum of degenerate terms is bounded as follows:
\begin{align*}
  |V_N| \le  \Bigg|\sum_{s = 1}^{N - 1} &\sum_{d = 1}^{ {s - 1 \over 2}} \sum_{m + n + p = s - 1 -d} C { {s - 1 \over 2} \choose d}^2 \Phi^{s + 2} {(b - a)^{s - 1} \over m! n! p! N^d} \sum_{q = 0}^s (-1)^q \nu^{s - 2 q}\Bigg| \\
  &\le \sum_{s = 1}^{N - 1} \sum_{d = 1}^{ {s - 1 \over 2}} \sum_{m + n + p = s - 1 - d} C { {s - 1 \over 2} \choose d}^2 \Phi^{s + 2} {(3 (b - a))^{s - 1} \over (s - 1)!} s N^{- d}\\
  &\le \sum_{d = 1}^{ {N - 2 \over 2} } \sum_{s = 2 d + 2}^{N - 1} C {(6 \Phi (b - a))^{s - 1} s \over (s - 1)!} {s - d + 2 \choose 2}\\
  &\le \sum_{d = 1}^{ {N - 2 \over 2}} \sum_{s = 1}^{N - 1} C{(12 \Phi (b - a))^{s - 1} s \over (s - 1)!} \le C(N^{-1} + \sum_{d = 2}^{ {N - 2 \over 2}} N^{-2}) \le C N^{-1}.
\end{align*}
\end{proof}

\section{The unitarity of $W$}\label{s:unitarity}

Let $f$ and $g$ be the generating functions of $D$ and $E$:
\begin{align*}
  f := \sum_{m, n, p \ge o} \sum_{q = 0}^{m + n + p + 1} D_{m, n, p; q} [m, n, p] (-\bar \nu)^q \nu^{m + n + p + 1 - q} = \sum_{s \ge 1} f_s ,\\
  g := \sum_{m, n, p \ge o} \sum_{q = 0}^{m + n + p + 1} E_{m, n, p; q} [m, n, p]^\dagger (-\bar \nu)^q \nu^{m + n + p + 1 - q} = \sum_{s \ge 1} g_s.\\
\end{align*}
And the kernel of $W - I$ is 
\begin{align}
  f(x, y) <_a^b(x, y) + g(x, y) >_a^b(x, y).
  \label{eq:wfg}
\end{align}
One can write down the equation that $f$ and $g$ have to satisfy for $W$ to be unitary.

\begin{proposition}
  For $W$ to be unitary, it suffices to show that for any $a < x < y < b$,
  \begin{align}
    f(x, y) + \overline{g(y, x)} + \int_a^x g(x, z) \overline{g(y, z)} dz + \int_x^y f(x, z) \overline{g(y, z)} dz + \int_y^b f(x, z) \overline{f(y, z)} =0.
    \label{eq:int}
  \end{align}
\end{proposition}

\begin{proof}
For $W$ to be unitary it is necessary and sufficient to show it is both a coisometry and an isometry%
\begin{eqnarray*}
  W^*W = W W^* = I
\end{eqnarray*}%
Plugging in \eqref{eq:wfg} and using the formulas for kernels of products and adjoints of integral operators we obtain equation \eqref{eq:int} and three ``other'' equations:
\begin{align}
  \overline{f(y, x)} + g(x, y) + \int_a^y g(x, z) \overline{g(y, z)} dz + \int_y^x g(x, z) \overline{f(y, z)} dz \notag\\
  + \int_x^b f(x, z) \overline{f(y, z)} dz = 0, \qquad a < y < x < b \label{eq:int2}\\
  \overline{g(y, x)} + f(x, y) + \int_a^x f(z, y) \overline{f(z, x)} dz + \int_x^y f(z, y) \overline{g(z, x)} dz \notag\\
  + \int_y^b g(z, y) \overline{g(z, x)} dz = 0, \qquad a < x < y < b \label{eq:int3}\\
  \overline{f(y, x)} + g(x, y) + \int_a^y f(z, y) \overline{f(z, x)} dz + \int_y^x g(z, y) \overline{f(z, x)} dz \notag\\
  + \int_x^b g(z, y) \overline{g(z, x)} dz = 0, \qquad a < y < x < b \label{eq:int4}
\end{align}
The equation \eqref{eq:int} and \eqref{eq:int2} are equivalent: one can interchange $x$ with $y$ and take a conjugate in the former to obtain the latter. So are \eqref{eq:int3} and \eqref{eq:int4}.
By the symmetry of $f$ and $g$ from Corollary \ref{c:symmetry}, \eqref{eq:int} and \eqref{eq:int3} are equivalent, hence it suffices to verify \eqref{eq:int} to show the unitarity of $W$.
\end{proof}

\begin{proof}[Proof of Theorem \ref{t:unitarity}]
We divide the proof into two parts: first we write down $f$ and $g$ in a more amenable form, then we proceed to proving the integral identity.
\subsection{The formulas for $f$ and $g$}
As previous, without loss of generality suppose $|\nu| = 1$. 
 Recall that
\begin{align*}
B_j(x, y) = \sum_{n \ge 0}{(-1)^{n + j} x^{n + j} y^n \over (n + j)! n!} = (-1)^j (x / y)^{j / 2} J_j(2 \sqrt{x y}).
\end{align*}
Since $E_{m, n, p; q} = \mathbb I_{m = p = q, n = 0}$,
\begin{align}
g(x, y) = \nu \sum_{m \ge 0} (-1)^m {(y - a)^m (b - x)^m \over m! m!} = \nu B_0(y - a, b - x).
\label{eq:g}
\end{align}
Let $f_e$ and $f_o$ be the parts of the sum of $f$ where $m + n + p$ are even and odd respectively, i.e.
\begin{align*}
  f_e = \sum_{2 | m + n + p} \sum_{q = 0}^{m + n + p + 1} (x - a)^m (y - x)^n (b - y)^p (- \bar \nu)^q \nu^{m + n + p + 1 - q}\\
  f_o = \sum_{2 \nmid m + n + p} \sum_{q = 0}^{m + n + p + 1} (x - a)^m (y - x)^n (b - y)^p (- \bar \nu)^q \nu^{m + n + p + 1 - q}
\end{align*}

The function $f_e$ can be further divided into the ${n \choose q - m}$ part and the rest.
Let $u := (x - a) \nu$, $v := (y - x) \nu$, $w := (b - y) \nu$ and $z = - {\bar \nu \over \nu} = - \bar\nu^2$, then
\begin{equation}
\begin{aligned}
\nu \sum_{q \ge 0} &\sum_{m + n + p = 2 q} {u^m v^n w^p z^{q} \over m! n! p!} {n \choose q - m}\\
 &= \nu \sum_{q \ge 0} \sum_{m + n + p = 2 q} {u^m v^{2 q - m - p} w^p z^{q} \over m! (2 q - m - p)! p!} {2 q - m - p \choose q - m}\\
 &= \nu \sum_{q \ge 0} \sum_{0 \le m, p \le q} {u^m v^{2 q - m - p} w^p z^{q} \over m! (2 q - m - p)! p!} {2 q - m - p \choose q - m}\\
 &= \nu \sum_{q \ge 0} v^{2 q} z^{q} \sum_{m = 0}^q {(u / v)^m \over m! (q - m)!} \sum_{p = 0}^q {(w / v)^m \over m! (q - m)!}\\
 &= \nu \sum_{q \ge 0} {(u + v)^q (v + w)^q z^{q} \over q! q!} = \nu B_0(y - a, b - x). 
\end{aligned}
  \label{eq:femq}
\end{equation}
The rest of $f_e$ is slightly more complicated. Let $k = {m + n + p \over 2}$. We observe that
\begin{align*}
\sum_q D_{m, n, p; q} z^q - {n \choose k - m} z^k = - {n \choose k} z^k + {n \choose k} z^{k + 1} - {n \choose k + 1} z^{k + 1} + \\
\dots + {n \choose n - 1}z^n - {n \choose n} z^n + {n \choose n} z^{n + 1} = (z - 1)\sum_{q = k}^n {n \choose q} z^q. 
\end{align*}
Therefore the rest of $f_e$, i.e. the sum excluding the terms corresponding to ${n \choose k - m}$ is
\begin{equation}
  \begin{aligned}
\nu (z - 1) \sum_{k \ge 0} &\sum_{m + n + p = 2 k} {u^m v^n w^p \over m! n! p!} \sum_{q = k}^n {n \choose q} z^q\\
&= \nu (z - 1) \sum_{k \ge 0} \sum_{n = k}^{2 k} \sum_{q = k}^n {n \choose q} z^q {(u + w)^{2 k - n} \over (2 k - n)!} {v^n \over n!}\\
&= \nu (z - 1) \sum_{k \ge 0} \sum_{n = k}^{2 k} \sum_{q = k}^n {z^q \over q!}{(u + w)^{2 k - n} \over (2 k - n)!} {v^n \over (n - q)!}\\
&= \nu (z - 1) \sum_{k \ge 0} \sum_{n = 0}^k \sum_{q = k}^{k + n} {z^q \over q!} {(u + w)^{k - n} \over (k - n)!} {v^{k + n} \over (k + n -q)!}\\
&= \nu (z - 1) \sum_{k \ge 0} \sum_{n = 0}^k \sum_{q = 0}^n {z^{k + q} \over (k + q)!} {(u + w)^{k - n} \over (k - n)!} {v^{k + n} \over (n - q)!} \\
&= \nu (z - 1) \sum_{k \ge 0} \sum_{q = 0}^k \sum_{n = 0}^{k - q} {z^{k + q} \over (k + q)!} {(u + w)^{k - q - n} \over (k - q - n)!} {v^{k + q + n} \over n!}\\
&= \nu (z - 1) \sum_{k \ge 0} \sum_{q = 0}^k {(zv)^{k + q} \over (k + q)!} {(u + v + w)^{k - q} \over (k - q)!}\\
&= \nu (z - 1) \sum_{q \ge 0} \sum_{k \ge 0} {(zv)^{k + 2 q} \over (k + 2 q)!} {(u + v + w)^k \over k!} 
  \end{aligned}
  \label{eq:ferest}
\end{equation}
We keep \eqref{eq:ferest} to later merge it with a similar term in $f_o$.

For $f_o$, let $k = {m + n + p - 1 \over 2}$. We observe that
\begin{align*}
\sum_q D_{m, n, p; q} &z^q - {n \choose k - m} z^k \\
&= {n \choose k} z^{k + 1} - {n \choose k + 1} z^{k + 1} + \dots + {n \choose n - 1}z^n - {n \choose n} z^n + {n \choose n} z^{n + 1} \\
&= (z - 1)\sum_{q = k + 1}^n {n \choose q} z^q + {n \choose k} z^{k + 1}.
\end{align*}
Following the same procedure which leads \eqref{eq:ferest}, the sum corresponding to the first term in the RHS is
\begin{equation}
\nu (z - 1) \sum_{q \ge 0} \sum_{k \ge 0} {(zv)^{k + 2 q + 1} \over (k + 2 q + 1)!} {(u + v + w)^k \over k!} 
  \label{eq:fo1}
\end{equation}
whereas the contribution from the second term is computed as follows:
\begin{equation}
  \begin{aligned}
\nu \sum_{k \ge 0} &\sum_{m + n + p = 2 k + 1} {u^m v^n w^p \over m! n! p!} {n \choose k} z^{k + 1}\\
&= \nu \sum_{k \ge 0} \sum_{n = k}^{2 k + 1} {v^n \over n!} {(u + w)^{2 k + 1 - n} \over (2 k + 1 - n)!} {n \choose k} z^{k + 1}\\
&= \nu \sum_{k \ge 0} \sum_{n = 0}^{k + 1} {v^{n + k} \over n! k!} {(u + w)^{k + 1 - n} \over (k + 1 - n)!} z^{k + 1}\\
&= \nu \sum_{k \ge 0} {v^k \over k!} {z^{k + 1} (v + u + w)^{k + 1} \over (k + 1)!} = B_1(b - a, y - x). 
  \end{aligned}
  \label{eq:fo2}
\end{equation}

By summing up \eqref{eq:femq}, \eqref{eq:ferest}, \eqref{eq:fo1} \eqref{eq:fo2} and plugging in $u, v, w, z$ we obtain
\begin{align}
f(x, y) = \nu B_0(y - a, b - x) + B_1(b - a, y - x) - (\nu + \bar \nu) \sum_{q \ge 0} B_q(y - x, b - a) \bar\nu^q.
\label{eq:f}
\end{align}
Note that $\sum_{q \ge 0} B_q(y - x, b - a) \bar \nu^q = \sum_{q \ge 0} J_q(2 \sqrt{(b - a) (y - x)})\left( - \sqrt{ {y - x \over b - a}} \bar \nu \right)^q$ is a generating function of the Bessel functions.

\subsection{Verifying the identity \eqref{eq:int}}
We list a few useful properties of $B_j$ (where we let $B_{-1}(x, y) := - B_1(y, x)$):
\begin{enumerate}
\item $\partial_x B_j(x, y) = - B_{j - 1}(x, y), j \ge 0$,
\item $B_0(x, y) = B_0(y, x)$,
\item $B_j(0, y) = \delta_{j0}, j \ge 0$,
\item $\partial_y B_j(x, y) = B_{j + 1}(x, y)$.
\end{enumerate}
And an integral:
\begin{equation}
\begin{aligned}
\int_y^b &B_0(b - x, z - a)  B_j(z - y, b - a) dz = - \sum_{k \ge 0} B_k(b - x, z - a) B_{k + j + 1}(z - y, b - a)\bigg|_y^b\\
&= \begin{cases}
  - \sum_{k \ge 0} B_k(b - x, b - a) B_{k + j + 1}(b - y, b - a), & j \ge 0\\
  \sum_{k \ge 0}B_k(b - x, b - a) B_k(b - y, b - a) - B_0(b - x, y - a) & j = - 1.
\end{cases}
\end{aligned}
  \label{eq:b0bj}
\end{equation}

Substituting for $f$ from \eqref{eq:f} and $g$ from \eqref{eq:g} into \eqref{eq:int} gives
\begin{equation}
  \begin{aligned}
&B_1(b - a, y - x) + \int_a^b B_0(z - a, b - x) B_0(z - a, b - y) dz \\
&+ \int_y^b B_1(b - a, z - x) B_1(b - a, z - y) dz \\
&+ \nu B_0(y - a, b - x) + \bar \nu B_0(x - a, b - y) + \bar \nu \int_x^b B_1(b - a, z - x) B_0(z - a, b - y) dz \\
&+ \nu \int_y^b B_1(b - a, z - y) B_0(z - a, b - x) dz\\
&- (\nu + \bar \nu) \int_y^b \sum_{q \ge 0} B_1(b - a, z - x) B_q(z - y, b - a) \nu^q dz \\
&- (\nu + \bar \nu) \int_y^b \sum_{q \ge 0}B_1(b - a, z - y) B_q(z - x, b - a) \bar\nu^q dz\\
&- \bar\nu(\nu + \bar \nu) \int_x^b \sum_{q \ge 0} B_0(z - a, b - y) B_q(z - x, b - a) \bar\nu^q dz \\
&- \nu (\nu + \bar \nu) \int_y^b \sum_{q \ge 0} B_0(z - a, b - x) B_q(z - y, b - a) \nu^q dz \\
&+ (\nu + \bar\nu)^2 \int_y^b \sum_{q \ge 0} B_q(z - x, b- a) \bar \nu^q \sum_{q \ge 0} B_q(z - y, b - a) \nu^q dz \\
&- (\nu + \bar \nu) \sum_{q \ge 0}B_q(y - x, b - a) \bar\nu^q = 0.
  \end{aligned}
  \label{eq:bint}
\end{equation}

Let $G_j (x) = \sum_{k \ge 0} {(-1)^k x^k \over k! (k + 1)!}$, then for $\alpha > 0$ the following two integrals hold:
\begin{align*}
  &\int_0^x G_0(\alpha z) G_0(\beta z) dz = (\alpha - \beta)^{-1} (\alpha x G_1(\alpha x) G_0(\beta x) - \beta x G_1(\beta x) G_0(\alpha x)),\\
  &\int_0^z G_1(w) G_1(w + \alpha) dw = \alpha^{-1}(z G_1(z) G_0(z + \alpha) - (z + \alpha) G_1(z + \alpha) G_0(z)) + G_1(\alpha).
\end{align*}
The first of these two integrals is the well-known Lommel's integral, see e.g. Section 11 and 94 of \cite{bowman58}.
The second integral written in terms of an indefinite integral of the Bessel functions is
\begin{equation}
\begin{aligned}
  \int &{1 \over \sqrt{w^2 + \beta^2}} J_1(w) J_1\left(\sqrt{w^2 + \beta^2}\right) dw \\
  &= \beta^{-2} \left(w J_1(w) J_0\left(\sqrt{w^2 + \beta^2}\right) - \sqrt{w^2 + \beta^2} J_1\left(\sqrt{w^2 + \beta^2}\right) J_0(w)\right), \beta > 0.
\end{aligned}
  \label{eq:sg}
\end{equation}
This is a special case of the so-called Sonine-Gegenbauer type integral (see e.g. page 415 of \cite{watson95}).
However, the authors have not found an explicit formula like the one on the right hand side of \eqref{eq:sg} in the literature.

By these two integrals we have
\begin{align*}
  \int_a^b &B_0(z - a, b - x) B_0(z - a, b - y) dz \\
  &= \left((y - x)^{-1} \big( (b - y) B_1(b - a, b - y) B_0(b - a, b - x)\right. \\
  &\qquad\left.- (b - x) B_1(b - a, b - x) B_0(b - a, b - y) \right)\\
  \int_y^b &B_1(b - a, z - x) B_1(b - a, z - y) dz \\
  &= B_1(b - a, y - x) - (y - x)^{-1} \left( (b - y) B_1(b - a, b - y) B_0(b - a, b - x) \right. \\
  &\qquad \left.- (b - x) B_1(b - a, b - x) B_0(b - a, b - y) \right).
\end{align*}
Therefore the first and the second lines of \eqref{eq:bint} vanishes.

By the integral \eqref{eq:b0bj} (and interchanging $x$ and $y$ when necessary), the third and the fourth lines are reduced to their real part.

Now if the real part of $\nu$ is $0$ then we are done. Otherwise by subtracting the first and the second lines and the imaginary part of the third and the fourth lines from \eqref{eq:bint}, and dividing the remainder by the real part ${\nu + \bar \nu \over 2}$, we simplify the integral identity into
\begin{align*}
&B_0(y - a, b - x) + B_0(x - a, b - y) + \int_x^b B_1(b - a, z - x) B_0(z - a, b - y) dz \\
&+ \int_y^b B_1(b - a, z - y) B_0(z - a, b - x) dz - 2 \sum_{q \ge 0}B_q(y - x, b - a) \bar\nu^q \\
&- 2 \int_y^b \sum_{q \ge 0} B_1(b - a, z - x) B_q(z - y, b - a) \nu^q dz \\
&- 2 \int_y^b \sum_{q \ge 0}B_1(b - a, z - y) B_q(z - x, b - a) \bar\nu^q dz\\
&- 2 \bar\nu \int_x^b \sum_{q \ge 0} B_0(z - a, b - y) B_q(z - x, b - a) \bar\nu^q dz \\
&- 2 \nu \int_y^b \sum_{q \ge 0} B_0(z - a, b - x) B_q(z - y, b - a) \nu^q dz \\
&+ 2 (\nu + \bar\nu) \int_y^b \sum_{q \ge 0} B_q(z - x, b- a) \bar \nu^q \sum_{q \ge 0} B_q(z - y, b - a) \nu^q dz = 0
\end{align*}

By the integral formulas \eqref{eq:b0bj}, the above identity can be further simplified to
\begin{align*}
& \sum_{k \ge 0} B_k(b - x, b - a) B_k(b - y, b - a) + \sum_{q > 0}\sum_{k \ge 0} B_k(b - x, b - a) B_{k + q}(b - y, b - a) \nu^q \\
&+ \sum_{q > 0} \sum_{k\ge 0} B_k(b - y, b - a) B_{k + q}(B - x, b - a) \bar\nu^q - \sum_{q \ge 0}B_q(y - x, b - a) \bar\nu^q \\
&- \int_y^b \sum_{q \ge 0} B_1(b - a, z - x) B_q(z - y, b - a) \nu^q dz \\
&- \int_y^b \sum_{q \ge 0}B_1(b - a, z - y) B_q(z - x, b - a) \bar\nu^q dz\\
&+ (\nu + \bar\nu) \int_y^b \sum_{q \ge 0} B_q(z - x, b- a) \bar \nu^q \sum_{q \ge 0} B_q(z - y, b - a) \nu^q dz = 0
\end{align*}
It remains to verify the coefficient of $\nu^q$ for each $q \in \mathbb Z$ in the LHS is 0, which can be done by repeated use of integration by parts and the properties of the $B_j$ functions.

When $q = 0$, the coefficient is
\begin{align*}
&\sum_{j \ge 0}B_j(b - x, b - a) B_j(b - y, b - a) - B_0(y - x, b - a) \\
&+ \int_y^b - B_1(b - a, z - y) B_0(z - x, b - a) - B_1(b - a, z - x) B_0(z - y, b - a) \\
&+ \sum_{j \ge 0} B_{j + 1}(z - x, b - a) B_j(z - y, b - a) + \sum_{j \ge 0} B_{j + 1}(z - y, b - a) B_j(z - x, b - a) dz \\
& = \sum_{j \ge 0}B_j(b - x, b - a) B_j(b - y, b - a) - B_0(y - x, b - a) \\
&+ \sum_{j \ge -1} \int_y^b B_j(z - x, b - a) B_{j + 1}(z - y, b - a) + B_j(z - y, b - a) B_{j + 1}(z - z, b - a) dz\\
& = \sum_{j \ge 0}B_j(b - x, b - a) B_j(b - y, b - a) - B_0(y - x, b - a) \\
&\qquad - \sum_{j \ge 0} B_j(z - x, b - a) B_j(z - y, b - a)\big|_y^b = 0,
\end{align*}
where in the fourth line we use integration by parts and the properties of $B_j$.

When $q > 0$, the coefficient is
{\small
\begin{align*}
&\sum_{k \ge 0} B_k(b - x, b - a) B_{k + q}(b - y, b - a) + \int_y^b - B_1(b - a, z - x) B_q(z - y, b - a) \\
&+ \sum_{j \ge 0}B_j(z - x, b - a) B_{j + q - 1}(z - y, b - a) + \sum_{j \ge 0} B_j(z - x, b - a) B_{j + q + 1} (z - y, b - a) dz\\
&= \sum_{k \ge 0} B_k(b - x, b - a) B_{k + q}(b - y, b - a) + \int_y^b \sum_{j \ge 0} B_j(z - x, b - a) B_{j + q - 1}(z - y, b - a) \\
&+ \sum_{j \ge -1} B_j(z - x, b - a) B_{j + q + 1} (z - y, b - a) dz\\
&= \sum_{k \ge 0} B_k(b - x, b - a) B_{k + q}(b - y, b - a) + \int_y^b \sum_{j \ge 0} B_j(z - x, b - a) B_{j + q - 1}(z - y, b - a) \\
&+ \sum_{j \ge 0} B_{j - 1}(z - x, b - a) B_{j + q} (z - y, b - a) dz\\
&= \sum_{k \ge 0} B_k(b - x, b - a) B_{k + q}(b - y, b - a) - \sum_{k \ge 0} B_k(z - x, b - a) B_{k + q}(z - y, b - a)\big|_y^b = 0.
\end{align*}
}
When $q < 0$, the coefficient is
{\small
\begin{align*}
&\sum_{k\ge 0} B_k(b - y, b - a) B_{k + q}(B - x, b - a) - B_q(y - x, b - a) \\
&+ \int_y^b - B_1(b - a, z - y) B_q(z - x, b - a) + \sum_{k \ge 0} B_{k + q - 1}(z - x, b - a) B_k(z - y, b - a) \\
&+ \sum_{k \ge 0} B_{k + q + 1}(z - x, b - a) B_k(z - y, b - a) dz\\
&= \sum_{k\ge 0} B_k(b - y, b - a) B_{k + q}(B - x, b - a) - B_q(y - x, b - a) \\
&+ \sum_{k \ge 0} \int_y^b B_{k + q - 1}(z - x, b - a) B_k(z - y, b - a) + B_{k + q}(z - x, b - a) B_{k - 1}(z - y, b - a) dz\\
&= \sum_{k\ge 0} B_k(b - y, b - a) B_{k + q}(B - x, b - a) - B_q(y - x, b - a) \\
&- \sum_{k \ge 0} B_{k + q}(z - x, b - a) B_k(z - y, b - a)\big|_y^b = 0.
\end{align*}
}
\end{proof}

\end{document}